\documentclass[a4paper,english]{lipics-v2018}

\usepackage{lipsum}

\usepackage{textcomp}
\usepackage{anyfontsize}

\usepackage{listings}
\lstdefinelanguage{program}{%
  keywords={%
    atomic,%
    assume,assert,call,return,new,%
    false,true,duplicate,restart,lock,unlock,%
    locate,insert,delete,contains,removeRight,rotateRightLeft,
    rcu_read_lock,rcu_read_unlock,synchronize_rcu%
  },  
  morecomment=[l]{//},
  morecomment=[s]{/*}{*/},
  morecomment=[n]{(**}{**)},
  mathescape=true,
  escapeinside=`',
}

\usepackage{booktabs}
\usepackage{mathpartir}
\usepackage{multirow}
\usepackage{amsmath}
\usepackage{amsthm}
\usepackage{amssymb}
\usepackage{MnSymbol}
\usepackage{stmaryrd}
\usepackage{tabularx}
\usepackage{wrapfig}
\usepackage{array}
\usepackage{subcaption}
\usepackage[T1]{fontenc}
\usepackage[nodayofweek]{datetime}
\usepackage{graphicx}
\usepackage{paralist}

\usepackage{xcolor}
\usepackage{extarrows}
\usepackage{mathtools}
\usepackage{paralist}
\usepackage{float}
\usepackage{array}
\usepackage{times}

\lstdefinestyle{lnumbers}{numbers=left,  stepnumber=1, numberblanklines=false, numberstyle=\tiny,basicstyle=\scriptsize, numbersep=3pt, escapeinside={/*}{*/}}
\lstdefinestyle{nonumbers}{numbers=none, escapeinside={/*}{*/}}

\usepackage{cleveref} 

\newif\ifnitpick
\nitpickfalse

\newif\iftodobom
\todobomtrue

\newif\iflong
\longtrue

\theoremstyle{remark}
\newtheorem{remarknum}{Remark}

\Crefname{conjecture}{Conjecture}{Conjectures}
\Crefname{proposition}{Proposition}{Propositions}
\Crefname{lemma}{Lemma}{Lemmas}
\Crefname{corollary}{Corollary}{Corollaries}
\Crefname{example}{Example}{Examples}
\Crefname{definition}{Definition}{Definitions}
\Crefname{figure}{Fig.}{Fig.}
\Crefname{section}{Sec.}{Sec.}

\nolinenumbers

\title{Order out of Chaos: Proving Linearizability Using Local Views} 

\titlerunning{Proving Linearizability Using Local Views}

\author{Yotam M.\ Y.\ Feldman}{Tel Aviv University, Israel}{}{}{}
\author{Constantin Enea}{IRIF, Univ.\ Paris Diderot \& CNRS, France}{}{}{}
\author{Adam Morrison}{Tel Aviv University, Israel}{}{}{}
\author{Noam Rinetzky}{Tel Aviv University, Israel}{}{}{}
\author{Sharon Shoham}{Tel Aviv University, Israel}{}{}{}

\authorrunning{YMY Feldman et al.}

\Copyright{Yotam M.\ Y.\ Feldman, Constantin Enea, Adam Morrison, Noam Rinetzky and Sharon Shoham} 

\subjclass{}


\iflong
\else
\relatedversion{The extended version of this paper appears in \href{https://arxiv.org/abs/1805.03992}{https://arxiv.org/abs/1805.03992}.}
\fi

\iflong
\ArticleNo{23}
\hideLIPIcs  
\fi

\iflong
\else
\EventEditors{Ulrich Schmid and Josef Widder}
\EventNoEds{2}
\EventLongTitle{32nd International Symposium on Distributed Computing (DISC 2018)}
\EventShortTitle{DISC 2018}
\EventAcronym{DISC}
\EventYear{2018}
\EventDate{October 15--19, 2018}
\EventLocation{New Orleans, USA}
\EventLogo{}
\SeriesVolume{121}
\ArticleNo{23} 
\fi

\newcommand{\eqdef}{\mathbin{\stackrel{\scriptscriptstyle{\textrm{def}}}{\boldsymbol{=}}}}

\newcommand{\totalto}{\rightarrow}

\newcommand{\ignore}[1]{}
\newcommand{\COMMENT}[1]{}

\newcommand{\citrus}[1]{}

\newif\ifcomments
\commentsfalse 
\ifcomments\usepackage[backgroundcolor=yellow,textsize=scriptsize]{todonotes}
\else
\usepackage[disable]{todonotes}
\fi

\ifcomments

\newcommand{\noamtodo}[1]{\textblue{#1}}
\newcommand{\sharon}[1]{\textblue{SH: #1 }}

\newcommand{\yf}[1]{\textred{YF: #1 }}
\newcommand{\textred}[1]{\color{red}#1\color{black}}
\newcommand{\textblue}[1]{\color{blue}#1\color{black}}

\else

\newcommand{\noamtodo}[1]{}
\newcommand{\sharon}[1]{}
\newcommand{\yf}[1]{}
\newcommand{\textred}[1]{#1}
\newcommand{\textblue}[1]{#1}

\fi

\newcommand{\TODO}[1]{\todo[inline]{\textbf{TODO: #1}}}

\newcommand{\noamdone}[1]{}

\newcommand{\noampardone}[1]{}
\newcommand{\noamdropped}[1]{}

\newcommand{\sharondone}[1]{}
\newcommand{\sharonx}[1]{}

\newcommand{\costin}[1]{\todo[inline]{CE: #1}}
\newcommand{\costindone}[1]{}

\newcommand{\yotam}[1]{\todo[inline]{YF: #1}}

\newcommand{\yfx}[1]{}%

\newcommand{\DONE}[1]{}

\ifnitpick
\newcommand{\nitpick}[1]{#1}
\else
\newcommand{\nitpick}[1]{}
\fi

\iftodobom
\newcommand{\todobom}[1]{\TODO{\textblue{#1}}}
\else
\newcommand{\todobom}[1]{}
\fi 

\newcommand{\imark}{\mathit{mark}}

\newcommand{\slheadobj}{{\tt head}}

\newcommand{\slQReachXK}[2]{\slheadobj\searchpath{#2}{#1}}

\newcommand{\slnextnode}{\operatorname{nextNode}}

\newcommand{\ttl}{\code{topLevel}}

\newcommand{\tk}{\code{key}}
\newcommand{\tn}{\code{ref}}
\newcommand{\tref}{\code{ref}}
\newcommand{\tm}{\code{mark}}
\newcommand{\tmnext}{\code{next}}

\newcommand{\lflI}[1]{I_{#1}}
\newcommand{\lflTI}[1]{\delta_{#1}}

\newcommand{\slI}[1]{I_{#1}}

\newcommand{\lflheadobj}{{\tt head}}

\newcommand{\lflQReachXYK}[3]{#1\searchpath{#3}{#2}}

\newcommand{\lflQReachXK}[2]{\lflheadobj\searchpath{#2}{#1}}

\newcommand{\inext}{\mathit{next}}

\newcommand{\PredXF}{\varphi}

\newcommand{\QXM}[1]{#1.\mathit{mark}}

\newcommand{\para}[1]{\vspace{0.3cm}\noindent \textbf{#1.}}

\newcommand{\nextChild}{\operatorname{nextChild}}
\newcommand{\ikey}{\mathit{key}}
\newcommand{\ileft}{\mathit{left}}
\newcommand{\iright}{\mathit{right}}
\newcommand{\repfunc}{\mathcal{A}}

\newcommand{\QReachXYK}[3]{#1\searchpath{#3}{#2}}

\newcommand{\PReachXK}[2]{\diamondminus(\QReachXK{#1}{#2})}
\newcommand{\QReachXK}[2]{\rootobj\searchpath{#2}{#1}}

\newcommand{\QXK}[3][=]{#2.\mathit{key}#1#3}
\newcommand{\QXKR}[3][=]{#3#1#2.\mathit{key}}
\newcommand{\QXD}[2][]{#1#2.\mathit{del}}

\newcommand{\QXR}[2][]{#1#2.\mathit{rem}}

\newcommand{\PastSymb}{\diamondminus}
\newcommand{\Past}[1]{\diamondminus (#1)}

\newcommand{\fontCONST}[1]{\mathit{#1}}

\newcommand{\seqr}{\bar{r}}
\newcommand{\seqw}{\bar{w}}
\newcommand{\seqwf}{\bar{w}_f}

\newcommand{\NULL}{\fontCONST{null}}

\newcommand{\loc}{\ell}

\newcommand{\noamxr}[1]{}

\newcommand{\polC}[1][]{\mathrm{C}}
\newcommand{\polphi}[1][]{\phi}

\newcommand{\polCsyn}[1][]{\mathtt{C}}
\newcommand{\polphisyn}[1][]{\boldsymbol{\phi}}

\newcommand{\code}[1]{\ensuremath{\mathtt{#1}}}

\newcommand{\wlupdate}[1][]{\operatorname{upd}}
\newcommand{\wlupdateset}[1][]{\operatorname{updSets}}

\newcommand{\readset}[1]{\mathit{ReadSet}_{#1}}

\newcommand{\Pred}{\mathbb{P}}
\newcommand{\PredOv}{\Pred}

\newcommand{\hlocal}{H_{\mathit{lv}}}
\newcommand{\hlocali}[1]{H^{(#1)}_{\mathit{lv}}}
\newcommand{\hinit}{H_c^{(0)}}
\newcommand{\hobserved}{H_{f}}
\newcommand{\hconc}{H_{c}}

\newcommand{\opprefix}[2]{#1^{({#2})}}

\newcommand{\leqloc}{\leq}
\newcommand{\leqacculoc}{\leq^{\cup}}

\newcommand{\Prew}[1]{\mathbb{Q}_{#1}}

\newcommand{\dataf}{{\tt data}}
\newcommand{\fdel}{{\tt del}}

\newcommand{\false}{\textit{false}}
\newcommand{\true}{\textit{true}}

\newcommand{\rootobj}{{\tt root}}

\newcommand{\searchpath}[1]{\overset{#1}{\leadsto}}

\newcommand{\readr}[1]{\mathit{read}(#1)}
\newcommand{\modw}[1]{\mathop{\textit{mod}}({#1})}

\newcommand{\validw}[1]{#1}

\newcommand{\PredSearchSymb}{\mathbb{S}}
\newcommand{\PredSearch}[2]{\PredSearchSymb_{#1,#2}}

\newcommand{\predset}{\mathcal{PR}}

\iflong
\newcommand{\refappendix}[1]{\Cref{#1}}
\else
\newcommand{\refappendix}[1]{the extended version~\cite{extendedVersion}}
\fi

\begin{document}

\maketitle

\begin{abstract}
Proving the linearizability of highly concurrent data structures, such as those using optimistic concurrency control, is a challenging task. The main difficulty is in reasoning about the view of the memory obtained by the threads, because as they execute, threads observe different fragments of memory from different points in time. Until today, every linearizability proof has tackled this challenge from scratch.

We present a unifying proof argument for the correctness of unsynchronized traversals, and apply it to prove the linearizability of several highly concurrent search data structures, including an optimistic self-balancing binary search tree, the Lazy List and a lock-free skip list.
Our framework harnesses {\em sequential reasoning} about the view of a thread, considering the thread as if it traverses the data structure without interference from other operations. 
Our key contribution is showing that properties of reachability along search paths can be deduced for concurrent traversals from such interference-free traversals, when certain intuitive conditions are met. Basing the correctness of traversals on such \emph{local view arguments} greatly simplifies linearizability proofs. 
At the heart of our result lies a notion of \emph{order on the memory}, corresponding to the order in which locations in memory are read by the threads, which guarantees a certain notion of consistency between the view of the thread and the actual memory.

To apply our framework, the user proves that the data structure satisfies two conditions: (1) acyclicity of the order on memory, even when it is considered across intermediate memory states, and (2) preservation of search paths to locations modified by interfering writes.
Establishing the conditions, as well as the full linearizability proof utilizing our proof argument, reduces to simple concurrent reasoning.
The result is a clear and comprehensible correctness proof, and elucidates common patterns underlying several existing data structures.
\end{abstract}


\clearpage


\section{Introduction}\label{Se:Intro}

Concurrent data structures must minimize synchronization to obtain high performance~\cite{David:2015:ACS,TAOMPP}.
Many concurrent search data structures therefore use {\em optimistic} designs, which search the data structure without
locking or otherwise writing to memory, and write to shared memory only when modifying the data structure.
Thus, in these designs, operations that do not modify the same nodes do not synchronize with each
other; in particular, searches can run in parallel, allowing for high performance and scalability.
Optimistic designs are now common in concurrent search trees~\cite{Arbel:2014,Brown:2014:GTN,Clements:2012,PFL16:CFTree,Drachsler:2014,Ellen:2010,Howley:2012,Natarajan:2014,intlf},
skip lists~\cite{NoHotSpotSkipList,fraser-phd,OptSkipList}, and lists/hash tables~\cite{HarrisList,HellerHLMMS05,MichaelList,Triplett:2011:RSC}.

A major challenge in developing an optimistic search data structure is proving
{\em linearizability}~\cite{TOPLAS:HW90}, i.e., that every operation appears to take effect atomically
at some point in time during its execution.  Usually, the key difficulty is 
proving properties of unsynchronized searches~\cite{PODC10:Hindsight,DBLP:conf/wdag/Lev-AriCK15,PPOPP:VafeiadisHHS06,TAOMPP}, as they can
observe an {\em inconsistent} state of the data structure---for example, due to observing only some of
the writes performed by an update operation, or only some update operations but not others.
Arguing about such searches requires tricky {\em concurrent reasoning} about the possible interleaving of
reads and writes of the operations.
Today, every new linearizability proof tackles these problems from scratch, leading to long and
complex proofs.

\para{Our approach: local view arguments} 
This paper presents a unifying proof argument for proving linearizability of concurrent data structures
with unsynchronized searches that replaces the difficult concurrent reasoning described above
with {\em sequential reasoning} about a search, which does not consider interference from other operations.
Our main contribution is a framework for establishing properties of an unsynchronized search in a concurrent
execution by reasoning {\em only} about its {\em local view}---the (potentially inconsistent) picture of memory it
observes as it traverses the data structure.
We refer to such proofs as \emph{local view arguments}.
We show that under two (widely-applicable) conditions listed below,
the existence of a path to the searched node in the local view, 
deduced with sequential reasoning,
also holds at some point
during the \emph{actual} (concurrent) execution of the traversal. 
(This includes the case of non-existence of a key indicated by a path to \code{null}.)
Such {\em reachability} properties are typically key to the linearizability proofs of many prominent concurrent search data structures
with unsynchronized searches~\cite{David:2015:ACS}.
Once these properties are established, the rest of the linearizability proof requires only simple concurrent reasoning.


%

Applying a local view argument requires establishing the following two conditions:
\begin{inparaenum}[(i)]
	\item \emph{temporal acyclicity}, which states that the search follows an \emph{order} on the memory that is \emph{acyclic} across intermediate states throughout the concurrent execution; and
    \item \emph{preservation}, which states that
	whenever a node $x$ is changed, if it was on a search path for some key $k$ in the past, then it is also on such a search path at the time of the change.
\end{inparaenum}
Although these conditions refer to concurrent executions, proving them for the data structures we consider is straightforward.

More generally, these conditions can be established with inductive proofs that are simplified by relying
on the very same 
traversal properties obtained with the local view argument.
This seemingly circular reasoning holds because our framework is also proven inductively,
and so the case of executions of length $N+1$ in both the proof that (1) the data structure satisfies the
conditions and (2) the traversal properties follow from the local view argument can rely on the correctness of the
other proof's $N$ case.

\para{Simplifying linearizability proofs with local view arguments}
To harness local view arguments, our approach uses \emph{assertions} in the code 
as a way to divide the proof between (1) the linearizability proof that \emph{relies on the assertions}, 
and (2) the proof of the assertions, where the challenge of establishing properties of unsynchronized searches in concurrent executions is overcome by local view arguments. 
%
%
\yotam{Omitted: \\
We illustrate this approach with a proof of the linearizability of a self-balancing binary search tree (\Cref{Se:Motivating}), using local view arguments (\Cref{Se:LocalViewArguments}).
} 

Overall, our proof argument yields clear and comprehensible linearizability proofs, whose whole is
(in some sense) greater than the sum of the parts, since each of the parts
requires a simpler form of reasoning compared to contemporary linearizability proofs.
We use local view arguments to devise simple linearizability proofs of a variant of the contention-friendly
tree~\cite{PFL16:CFTree} (a self-balancing search tree)\citrus{ and the Citrus tree~\cite{Arbel:2014} (an RCU-based
search tree),}{}, lists with lazy~\cite{HellerHLMMS05} or non-blocking~\cite{TAOMPP} synchronization, and a lock-free skip list.


Our framework's \emph{acyclicity} and \emph{preservation} conditions can provide insight on algorithm
design, in that their proofs can reveal unnecessary protections against interference.
Indeed, our proof attempts exposed (small) parts of the search tree algorithm that were not needed to guarantee linearizability, leading us to consider a simpler variant of its search operation (see \Cref{rem:CF-diff}).


\para{Contributions}
To summarize, we make the following contributions:
\begin{enumerate}
	\item We provide a set of conditions under which reachability properties of local views, established using sequential reasoning, hold also for concurrent executions,
	\item We show that these conditions hold for non-trivial concurrent data structures that use unsynchronized searches, and
	\item We demonstrate that the properties established using local view arguments enable simple linearizability proofs, 
	alleviating the need to consider interleavings of reads and writes during searches.
\end{enumerate}

\newcommand{\clearpagex}{}

\section{Motivating Example}\label{Se:Running}\label{Se:Motivating}



\lstset{language={program},style=lnumbers,firstnumber=last,basicstyle=\ttfamily\scriptsize,numberstyle=\tiny,escapeinside={(**}{**)},backgroundcolor=\color{white}}

\begin{figure}[t]
\centering
\begin{tabular}[t]{p{5.6cm}p{5.0cm}p{4cm}}
\begin{lstlisting}
type N 
  int key  
  N left, right 
  bool del,rem

N root$\leftarrow$new N($\infty$);

N$\times$N locate(int k)
  x,y$\leftarrow$root
  while (y$\neq$null $\land$ y$.$key$\neq$k)
    x$\leftarrow$y
    if (x$.$key<k)
      y$\leftarrow$x$.$right
    else
      y$\leftarrow$x$.$left
 $
  \begin{array}{l}
  \{\PReachXK{x}{k}  \land \PReachXK{y}{k}   \\
    \ \ \land\ \QXK[\neq]{x}{k} \land y \neq \NULL \implies  \QXK{y}{k}\}
  \end{array}
  \label{Ln:RetLocate}
  $  
  return (x,y)

bool contains(int k)  
  (_,y)$\leftarrow$locate(k) 
  if (y = null) 
    $  
    \{ \Past{\QReachXK{\NULL}{k}} \} \label{Ln:ContainsRetFalseNULL}
    $
    return false
  $\{ \PReachXK{y}{k} \}  \label{Ln:ContainsBeforeReadDEL}$
  if (y.del)
    $  
    \{ \Past{\QReachXK{y}{k} \land \QXD{y}} \land \QXK{y}{k} \} \label{Ln:ContainsRetFalseDEL}
    $
    return false
  $  
    \{ \Past{\QReachXK{y}{k} \land \neg \QXD{y}}  \land  \QXK{y}{k}  \} \label{Ln:ContainsRetTrue}
  $ 
  return true  
\end{lstlisting}



&




\begin{lstlisting}
bool delete(int k) 
  (_,y)$\leftarrow$locate(k) 
  if (y = null)
    $  
      \{ \Past{\QReachXK{\NULL}{k}} \} \label{Ln:DeleteYNull}
    $  
    return false $\label{Ln:DeleteRetFalseNull}$
  lock(y)
  if (y$.$rem) restart
  ret $\leftarrow$ $\neg$y$.$del
  $
       \{\QReachXK{y}{k} \land \QXK[=]{y}{k} \land \QXR[\neg]{y} \} \label{Ln:DeleteSetDel}% \land \mathit{ret} = \neg y.\mathit{del}\}
  $
  y$.$del$\leftarrow$true
  return ret                $\label{Ln:DeleteRetNotDel}$

bool insert(int k)  
  (x,y)$\leftarrow$locate(k)
  $
  %\{\diamondminus(\QReachXK{x}{k}\land \QXK[\neq]{x}{k} \} \label{Ln:InsertLocate}
  \{\PReachXK{x}{k}\land \QXK[\neq]{x}{k} \} \label{Ln:InsertLocate}
  $  
  if (y$\neq$null)
    $
        \{\PReachXK{y}{k}\land \QXK[=]{y}{k}\} 
    $
    lock(y)
    if (y$.$rem) restart
    ret $\leftarrow$ y$.$del
    $
        \{\QReachXK{y}{k}\land \QXK[=]{y}{k} \land \QXR[\neg]{y}\} \label{Ln:InsertSetDel} 
    $
    y.del$\leftarrow$false 
    return ret $ \label{Ln:InsertRetDel} $
  lock(x)
  if (x$.$rem) restart
  if (k < x$.$key  $\land$ x.left$=$null)
    $
    \begin{array}{l}
    \{\QReachXK{x}{k} \land \QXR[\neg]{x}   \\
      \ \ \land\  k < x.key \land x.\mathit{left}=\text{null} \}
    \end{array}
    \label{Ln:InsertInsertLeft}
    $
    x.left $\leftarrow$ new N(k)
  else if (x.right$=$null)
    $
    \begin{array}{l}
    \{\QReachXK{x}{k} \land \QXR[\neg]{x}   \\
      \ \ \land\  k > x.key \land x.\mathit{right}=\text{null} \}
    \end{array}
    \label{Ln:InsertInsertRight}
    $         
    x.right $\leftarrow$ new N(k)  
  else 
    restart
  return true
\end{lstlisting}




&



\begin{lstlisting}
removeRight()  
  (z,_) $\leftarrow$ locate(*)
  lock(z)
  y $\leftarrow$ z.right
  if(y=null $\lor$ z.rem) 
    return
  lock(y)
  if (y.del)
    return
  if (y.left$=$null)
    $\label{Ln:RemoveRight} $ z.right $\leftarrow$ y.right
  else if (y.right$=$null)
    $\label{Ln:RemoveLeft} $ $\label{cft-code:remove-bypass}$z.right $\leftarrow$ y.left
  else
    return
  $\label{cft-code:remove-rem-y}$y.rem $\leftarrow$ true

rotateRightLeft()  
  (p,_) $\leftarrow$ locate(*)
  lock(p)
  y $\leftarrow$ p$.$left
  if(y=null $\lor$ p.rem) 
    return
  lock(y)
  x $\leftarrow$ y$.$left
  if(x=null)
    return
  lock(x)
  z $\leftarrow$ duplicate(y)
  z.left $\leftarrow$ x.right   $\label{Ln:RotateFresh}$ 
  x.right $\leftarrow$ z        $\label{Ln:RotateXRight}$  
  p.left $\leftarrow$ x         $\label{Ln:RotatePLeft}$  $\label{cft-code:rotate-change-parent}$
  y.rem $\leftarrow$ true       $\label{cft-code:rotate-rem-y}$
\end{lstlisting}
\end{tabular} 
\caption{\label{Fi:Running}
Running example. For brevity, \textbf{unlock} operations are omitted; a procedure releases all the locks it acquired when it terminates or \textbf{restart}s. $*$ denotes an arbitrary key.
}
\end{figure}

As a motivating example we consider a self-balancing binary search tree
with optimistic, read-only searches. This is an example of a concurrent data structure for which it is challenging to prove linearizability ``from scratch.'' The algorithm is based on the
contention-friendly (CF) tree~\cite{EuroPar13:CFTree,PFL16:CFTree}.  It
is a fine-grained lock-based implementation of a set object with the standard
\code{insert(}$k$\code{)}, \code{delete(}$k$\code{)}, and \code{contains}($k$\code{)} operations.
\yotam{I think we can do without footnote (although I asked about it...): For simplicity, we assume that keys are natural numbers.}
The algorithm maintains an {\em internal} binary tree that stores a key in every
node.
Similarly to the lazy list~\cite{HellerHLMMS05}, the algorithm distinguishes between
the {\em logical deletion} of a key, which removes it from the set represented by the tree, and the
{\em physical removal} that unlinks the node containing the key from the tree.

We use this algorithm as a running example to illustrate how our framework allows to lift sequential reasoning into assertions about concurrent executions, which are in turn used to prove linearizability. In this section, we present the algorithm and explain the linearizability proof based on the assertions, highlighting the significant role of local view arguments in the proof.

\Cref{Fi:Running} shows the code of the algorithm.
(The code is annotated with
assertions written inside curly braces, which the reader should ignore for
now; we explain them in~\Cref{Se:ProofLin}.)
Nodes contain two boolean fields, $del$ and $rem$, which indicate whether
the node is logically deleted and physically removed, respectively.  Modifications of
a node in the tree are synchronized with the node's lock.
Every operation starts with a call to \code{locate(k)}, which performs a standard
binary tree search---without acquiring any locks---to locate the node with the target
key $k$.
This method returns the last link it traverses, $(x,y)$.  Thus, if $k$ is found,
$y.key=k$; if $k$ is not found, $y=\NULL$ and $x$ is the node that would be $k$'s
parent if $k$ were inserted.
A \code{delete(k)} logically deletes $y$ after verifying that $y$ remained linked
to the tree after its lock was acquired.
An \code{insert(k)} either revives a logically deleted node or, if $k$ was not found,
links a new node to the tree.
A \code{contains(k)} returns $\true$ if it locates a node with key $k$
that is not logically deleted, and $\false$ otherwise.

Physical removal of nodes and balancing of the tree's height are performed using auxiliary
methods.\footnote{The reader should assume that these methods can be invoked at any time; the details
of when the algorithm decides to invoke them are not material for correctness.
For example, in~\cite{EuroPar13:CFTree,PFL16:CFTree}, these methods are invoked by a dedicated restructuring thread.}
\begin{figure}
\begin{center}
\begin{minipage}[t]{.65\textwidth}
\centering
\includegraphics[scale=0.5]{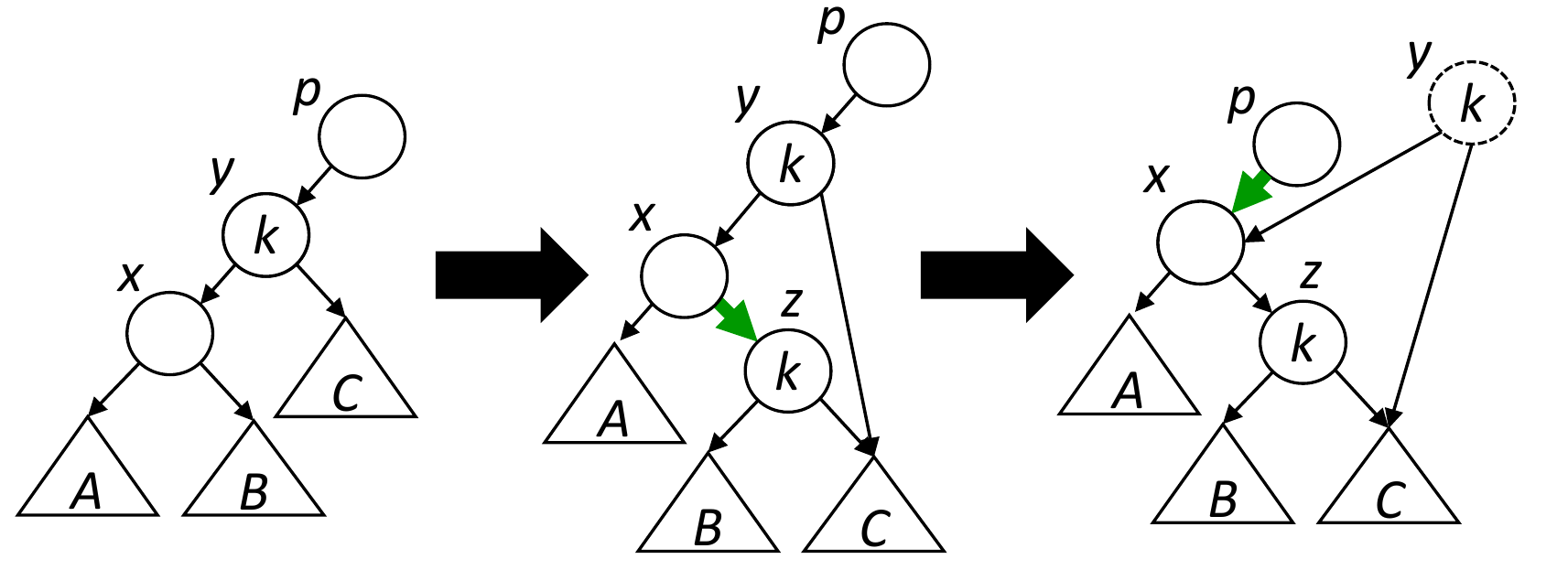}
\subcaption{\footnotesize Right rotation of $y$.  (The bold green link is the one written in each step.
The node with a dashed border has its $rem$ bit set.)
}
\label{Fi:Rotation}
\end{minipage}
\ \ \ 
\begin{minipage}[t]{.3\textwidth}
\centering
\includegraphics[scale=0.5]{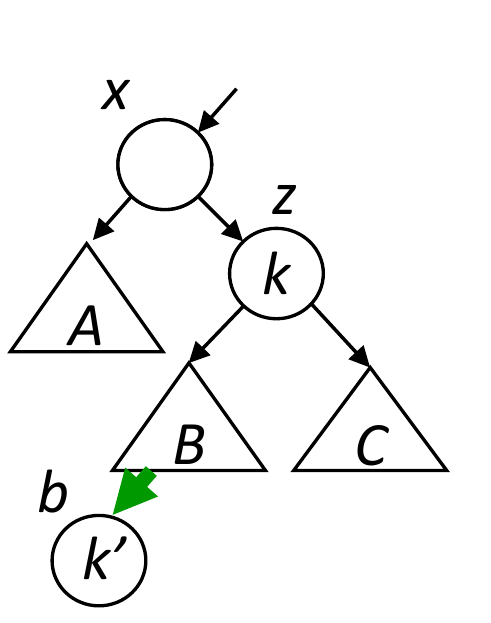}
\subcaption{\footnotesize Node $b$ is added after the right rotation of $y$, when $y$
is no longer in the tree.}
\label{Fi:RotationInconsist}
\end{minipage}
\end{center}
\caption{\footnotesize A right rotation, and how it can lead a search to observe an inconsistent state of the tree.
If $b$ is added after the rotation, a search for $k'$ that starts
before the rotation and pauses at $x$ during the rotation will traverse the path $p,y,x,z,\dots,b$,
although $y$ and $b$ never exist simultaneously in the tree.}
\label{Fi:RotationAndInconsist}
\end{figure}
%
%
%
%
%
%
The algorithm physically removes only nodes with at most one child.  The \code{removeRight}
method unlinks such a node that is a right child, and sets its $rem$ field to notify threads
that have reached the node of its removal.  (We omit the symmetric \code{removeLeft}.)
Balancing is done using rotations.
\Cref{Fi:Rotation} depicts the operation of \code{rotateRightLeft},
which needs to rotate node $y$ (with key $k$) down.
(We omit the     symmetric operations.)
It creates a new node $z$ with the same key and $\mathit{del}$ bit as $y$ to take $y$'s place, leaving $y$
unchanged except for having its $rem$ bit set.
A similar technique for rotations is used in lock-free search trees~\cite{Brown:2014:GTN}.


\begin{remarknum} \label{rem:CF-diff}
The example of \Cref{Fi:Running} differs from the original contention-friendly tree~\cite{EuroPar13:CFTree,PFL16:CFTree} in a few points.
The most notable difference is that our traversals do not consult the \code{rem} flag, and in particular we do not need to distinguish between a left and right rotate, making the traversals' logic simpler.
Checking the \code{rem} flag is in fact unnecessary for obtaining linearizability, but it allows proving linearizability with a \emph{fixed} linearization point, whereas
proving the correctness of the algorithm without this check requires an \emph{unfixed} linearization point.
For our framework, the necessity to use an unfixed linearization point incurs no additional complexity. In fact, the simplicity of our proof method allowed us to spot this ``optimization.''
 %
In addition, the original algorithm performs backtracking by setting pointers from child to parent when nodes are removed. Instead, we restart the operation;
see \Cref{Se:Conc} for a discussion of backtracking.
\yotam{Reference to future work sufficiently non-misleading but not too unsettling?}
Lastly, we fix a minor omission  in the description of~\cite{PFL16:CFTree}, where the \code{del} field was not copied from a rotated node.
\end{remarknum}

\costin{We should maybe reinforce the simplification "do not consult the \code{rem} flag" a little bit. Say that this was an unnecessary check, and that the simplicity of our proof method allowed us to spot this "optimization", even if it is minor. But more generally, this is a hint that an uniform and simple proof method can be beneficial even for designing concurrent data structures\\
SH: added something. is it ok?}

\costin{What do we mean by "we do not need to distinguish between a left and right rotate". And how does this simplify the logic ? It sounds mysterious to me.}


\subsection{Proving Linearizability}\label{Se:ProofLin}

Proving linearizability of an algorithm like ours is challenging because searches are performed with no synchronization. This means that, due to interference from concurrent updates,
searches may observe an inconsistent state of the tree that has not existed at any point in time.
(See \Cref{Fi:RotationAndInconsist}.)
In our example, while it is easy to see that \code{locate} in \Cref{Fi:Running} constructs a search path to a node in sequential executions, what this implies for concurrent traversals 
is not immediately apparent. Proving properties of the traversal---in particular, that a node reached in the traversal truly lies on a search path for key $k$---is instrumental for the linearizability proof~\cite{PPOPP:VafeiadisHHS06,PODC10:Hindsight}.

Generally, our linearizability proofs consist of two parts: (1) proving a set of \emph{assertions} in the code of the concurrent data structure, and (2) a proof of linearizability based on those assertions. The most difficult part and the main focus of our paper is proving the assertions using local view arguments, discussed in \Cref{Se:RunningProvingAssertions}.
In the remaining of this section we demonstrate that having assertions about the actual state during the concurrent execution makes it a straightforward exercise to verify that the algorithm in \Cref{Fi:Running} is a linearizable implementation of a set, {\em assuming these assertions}.

Consider the assertions in \Cref{Fi:Running}. An assertion  $\{\Pred\}$ means that $\Pred$ holds now (i.e., in any state in which the next line of code executes).
An assertion of the form $\{\PastSymb \Pred\}$ means that $\Pred$ was true at some point between the invocation of the
operation and now.
The assertions contain predicates about the state of locked nodes, immutable fields, and predicates of the form
$\QReachXK{x}{k}$, which means that $x$ resides on a \emph{valid search path} for
key $k$ that starts at $\rootobj$; if $x=\NULL$ this indicates that $k$ is not in the tree
(because a valid search path to $k$ does not continue past a node with key $k$).
Formally, search paths between objects (representing nodes in the tree) are defined as follows:
$$
\begin{array}{l}
\QReachXYK{o_r}{o_x}{k}
\eqdef
\exists o_0,\ldots,o_m.\,
o_0=o_r \land o_m=o_x \land
\forall i=1..m.\, \nextChild(o_{i-1},k,o_i)\,, \text{ and }
\\[2pt]
\qquad
\nextChild(o_{i-1},k,o_i) = (o_{i-1}.key > k \land o_{i-1}.\ileft = o_i) \lor (o_{i-1}.key < k \land o_{i-1}.\iright = o_i) \ .
\end{array}
$$

One can prove linearizability from these assertions by, for example, using an \emph{abstraction function}
$\repfunc: H \totalto \powerset(\mathbb{N})$ that maps a concrete memory state\footnote{We use standard modeling of the memory state (the {\em heap}) as a function $H$ from locations to values; see \Cref{Se:LocalViewArguments}.} of the tree, $H$, to
the {\em abstract set} represented by this state, and showing that \code{contains}, \code{insert}, and \code{delete}
manipulate this abstraction according to their specification.
We define $\repfunc$ to map $H$ to the set of keys of the nodes that are on a valid search path for their key and
are not logically deleted in $H$:
\yotam{Why are keys natural numbers?}
\yotam{Type of the domain of the function $\repfunc$}
\costin{I added footnotes 2 and 3 to makes these things precise}
$$
\repfunc(H) = \{ k \in \mathbb{N} \mid H \models \exists x.\,\QReachXK{x}{k} \land x.key=k \land \neg x.del\},\ \text{ where } H \models P \text{ means that } P \text{ is true in state }H.
$$
\TODO{Reviewer 1: using $H$ as a notation for state implies that the state is given by
a history. Is that the intent?}

The assertions almost immediately imply that for every operation invocation $op$, there exists a state $H$ during
$op$'s execution for which the abstract state $\repfunc(H)$ agrees with $op$'s return value, and so $op$ can
be linearized at $H$.
\sharonx{``Method'' vs.\ operation}%
\iflong
We need only make the following observations.  First, \code{contains} and a failed \code{delete}
or \code{insert} do not modify the memory, and so can be linearized at the point in time in which the assertions
before their \code{return} statements hold.  Second, in the state $H$ in which a successful \code{delete(k)}
(respectively, \code{insert(k)}) performs a write, the assertions on line~\ref{Ln:DeleteSetDel}
(respectively, lines \ref{Ln:InsertSetDel}, \ref{Ln:InsertInsertLeft}, and~\ref{Ln:InsertInsertRight}) imply that $k \in \repfunc(H)$
(respectively, $k \not\in \repfunc(H)$).  Therefore, these writes change the abstract set, making it agree
with the operation's return value of $\true$.  Finally, since these are the only memory modifications performed by the set
operations, it only remains to verify that no write performed by an auxiliary operation in state $H$ modifies $\repfunc(H)$.
Indeed, as an operation modifies a field of node $v$ only when it has $v$ locked,
it is easy to see that for any node $x$ and key $k$, if $\QReachXK{x}{k}$ held before the write, then it also
holds afterwards with the exception of the removed node $y$. However, \code{removeRight} removes a deleted node, and thus does not change $\repfunc(H)$.
Further, \code{rotateRightLeft} links $z$ ($y$'s replacement) to the tree before unlinking $y$, so the existence of a search path to $y.k=z.k$ is retained (although the actual path changes), leaving the contents of the abstract set unchanged
because the $\mathit{del}$ bit in $z$ has the same value as in~$y$.
\else
We provide a more detailed discussion in the extended version~\cite{extendedVersion}.
\fi

\subsection{Proving the Assertions}
\label{Se:RunningProvingAssertions}
To complete the linearizability proof, it remains to prove the validity of the assertions in concurrent executions.
The most challenging assertions to prove are those concerning properties of unsynchronized traversals, which we target in this paper.
%
In \Cref{Se:LocalViewArguments} we present our framework, which allows to deduce assertions of the form of $\PReachXK{x}{k}$ at the end of (concurrent) traversals by considering only interference-free executions. We apply our framework to establish the assertions $\PReachXK{x}{k}$ and $\PReachXK{y}{k}$ in~\cref{Ln:RetLocate}. 
In fact, our framework allows to deduce slightly stronger properties, namely, of the form $\Past{\QReachXK{x}{k} \land \varphi(x)}$, where $\varphi(x)$ is a property of a single field of $x$ (see \Cref{rem:ReachWithField}). This is used to prove the assertions $\Past{\QReachXK{y}{k} \land \QXD{y}}$ in \cref{Ln:ContainsRetFalseDEL} and similarly in \cref{Ln:ContainsRetTrue}.
%
For completeness, we now show how the proof of the remaining assertions in \Cref{Fi:Running} is attained, when assuming
the assertions deduced by the framework. This concludes the linearizablity proof.

\para{Reachability related assertions}
In \cref{Ln:ContainsBeforeReadDEL} the fact that $\PReachXK{y}{k}$ is true follows from~\cref{Ln:RetLocate}.

%

The writes in \code{insert} and \code{delete} (\cref{Ln:InsertSetDel,Ln:InsertInsertLeft,Ln:InsertInsertRight,Ln:DeleteSetDel}) require that a path exists \emph{now}.
This follows from the $\PReachXK{x}{k}$ (known from the local view argument) and the fact that $\QXR[\neg]{x}$, using an invariant similar to preservation (see \Cref{example:preservation-cf}):
For every location $x$ and key $k$, if $\QReachXK{x}{k}$, then every write retains this unless it sets $\QXR{x}$ before releasing the lock on $x$ (this happens in~\cref{Ln:RotatePLeft,Ln:RemoveLeft,Ln:RemoveRight}). Thus, when \code{insert} and \code{remove} lock $x$ and see that it is not marked as removed, $\QReachXK{x}{k}$ follows from $\PReachXK{x}{k}$.
Note that the fact that writes other than~\cref{Ln:RotatePLeft,Ln:RemoveLeft,Ln:RemoveRight} do not invalidate $\QReachXK{x}{k}$ follows easily from their annotations.

\para{Additional assertions}
The invariant that keys are immutable justifies assertions referring to keys of objects that are read earlier, e.g.\ in~\cref{Ln:InsertSetDel} and the rest of the assertion in~\cref{Ln:ContainsRetTrue}
($y.key$ is read earlier in \code{locate}).
The rest of the assertions can be attributed to reading a location under the protection of a lock. An example of this is the assertion that $\QXR[\neg]{y}$ in~\cref{Ln:DeleteSetDel}.

\section{The Framework: Correctness of Traversals Using Local Views}
\label{Se:LocalViewArguments}
In this section we present the key technical contribution of our framework, which targets proving properties of traversals.
We address properties of \emph{reachability along search paths} (formally defined in \Cref{Se:ProofPresentSearchPaths}).
Roughly speaking, our approach considers the traversal in concurrent executions as operating without interference on a \emph{local view}: the thread's potentially inconsistent picture of memory obtained by performing reads concurrently with writes by other threads.
For a property $\PredSearch{k}{x} =\QReachXK{x}{k}$ of reachability along a search path,
we introduce conditions under which one can deduce that $\PastSymb \PredSearch{k}{x}$
holds in the actual global state of the concurrent data structure out of the fact that $\PredSearch{k}{x}$ holds in the local view of a single thread, where the latter is established using sequential reasoning (see \Cref{Se:SeqReasoning}).
This alleviates the need to reason about intermediate states of the traversal in the concurrent proof.



This section is organized as follows:
We start with some preliminary definitions.
\Cref{Se:ProofPresentSearchPaths} defines the abstract, general notion of search paths our framework treats. 
\Cref{Se:local-view} defines the notion of a local view which is at the basis of local view arguments. 
\Cref{Se:FrameworkConditions} formally defines the conditions under which local view arguments hold, and states our main technical result.
In \Cref{Se:ProofIdea} we sketch the ideas behind the proof of this result.

\para{Programming model}
A \emph{global state} (state) is a mapping between \emph{memory locations} (locations)  and \emph{values}.
A value is either a natural number, a location, or   $\NULL$.
Without loss of generality, we assume that threads share access to a global state. Thus, memory locations are used to store the values of fields of objects.
A \emph{concurrent  execution} (execution) is a sequence of states produced by an interleaving of atomic actions issued by threads.
We assume that each atomic action is either a \emph{read} or a \emph{write} operation.
(We treat synchronization actions, e.g., \emph{lock} and \emph{unlock}, as writes.)
A \emph{read} $r$ consists of a value $v$ and a location $\readr{r}$ with the meaning that $r$ reads $v$ from $\readr{r}$.
Similarly, a write $w$ consists of a value $v$ and a location $\modw{w}$ with the meaning that $w$ sets $\modw{w}$ to $v$.
We denote by $w(H)$ the state resulting from the execution of $w$ on state $H$.

\subsection{Reachability Along Search Paths}
\label{Se:ProofPresentSearchPaths}

The properties we consider are given by predicates of the form $\PredSearch{k}{x} = \rootobj \searchpath{k} x$, denoting reachability of $x$ by a $k$-search path, where $\rootobj$ is the entry point to the data structure.
A $k$-search path in state $H$ is a sequence of locations 
that is traversed when searching for a certain element, parametrized by $k$, in the data structure.
Reachability of an \emph{object} $x$ along a $k$-search path from $\rootobj$ is understood as the existence of a $k$-search path between designated locations of $x$, e.g.\ the key field, and $\rootobj$.

Search paths may be defined differently in different data structures (e.g., list, tree or array).
For example, $k$-search paths in the tree of \Cref{Fi:Running} consist of sequences $\langle x.key, x.\ileft, y.key \rangle$ where $y.key$ is the address pointed to by $x.\ileft$ (meaning, the location that is the value stored in $x.\ileft$) and $x.key > k$, or
	$\langle x.key, x.\iright, y.key \rangle$ where $y.key$ is the address pointed to by $x.\iright$ and $x.key < k$.
	This definition of $k$-search paths reproduces the definition of reachability along search paths from \Cref{Se:ProofLin}.

Our framework is oblivious to the specific definition of search paths, and only assumes the following properties of search paths (which are satisfied, for example, by the definition above):
%
\begin{itemize}
		\item If $\ell_1,\ldots,\ell_m$ is a $k$-search path in $H$ and $H'$ satisfies $H'(\ell_i) = H(\ell_i)$ for all $1 \leq i < m$, then $\ell_1,\ldots,\ell_m$ is a $k$-search path in $H'$ as well, i.e., the search path depends on the values of locations in $H$ only for the locations along the sequence itself (but the last).
		\item If $\ell_1,\ldots,\ell_m$ and $\ell_m,\ldots,\ell_{m+r}$ are both $k$-search paths in $H$, then so is $\ell_1,\ldots,\ell_m,\ldots,\ell_{m+r}$, i.e., search paths are closed under concatenation.
		\item If $\ell_1,\ldots,\ell_m$ is a $k$-search path in $H$ then so is $\ell_i,\ldots,\ell_j$ for every $1 \leq i \leq j \leq m$, i.e., search paths are closed under truncation.
	\end{itemize}

\begin{remarknum}\label{rem:ReachWithField}
It is simple to extend our framework to deduce properties of the form $\Past{\QReachXK{x}{k} \land \PredXF(x)}$ where $\PredXF(x)$ is a property of a single field of $x$. For example, $\PredXF(x) = x.\mathit{del}$ states that the field $\mathit{del}$ of $x$ is true.
As another example, the predicate $\QReachXK{x}{k} \land (x.\mathit{next} = y)$
says that the \emph{link} from $x$ to $y$ is reachable. See \refappendix{Se:ReachWithField} for details.
\end{remarknum}

\subsection{Local Views and Their Properties}
\label{Se:local-view}
We now formalize the notion of \emph{local view} and explain how properties of local views can be established using sequential reasoning.

%

\para{Local view} Let $\seqr= r_1,\ldots,r_d$ be a sequence of 
read actions executed by some thread. As opposed to the global state, the \emph{local view} of the reading thread refers to the inconsistent picture of the memory state that the thread obtains after issuing $\seqr$ (concurrently with writes).
Formally, the sequence of reads $\seqr$ induces a state $\hlocal$, which is constructed by assigning to every location $x$ which $\seqr$ reads the last value $\seqr$ reads in $x$. Namely,
when $\seqr$ starts, its local view $\hlocali{0}$ is empty, and, assuming its $i$th read of value $v$ from location $\loc$, the
produced local view is $\hlocali{i}=\hlocali{i-1}[\loc \mapsto v]$.
We refer to $\hlocal=\hlocali{d}$ as the \emph{local view} produced by $\seqr$ (\emph{local view} for short).
We emphasize that while technically $\hlocal$ is a state, it is \emph{not} necessarily an actual intermediate global state, and may have never existed in memory during the execution.


\para{Sequential reasoning for establishing properties of local views} \label{Se:SeqReasoning}
Properties of the local view $\hlocal$, which are the starting point for applying our framework, are established using \emph{sequential} reasoning.
Namely, proving that a predicate such as $\QReachXK{x}{k}$ holds in the local view at the end of the traversal amounts to proving that it holds in  \emph{any sequential} execution of the traversal, i.e., an execution without interference which starts at an \emph{arbitrary} program state.  This is because the concurrent traversal constructing the local view can be understood as a sequential execution that starts with the local view as the program state.

%
\begin{example}
\label[example]{example:seq-reason-cf}
In the running example, straightforward sequential reasoning shows that indeed $\QReachXK{x}{k}$ holds at~\cref{Ln:RetLocate} in sequential executions of \code{locate(k)} (i.e., executions without interference), no matter at which program state the execution starts.
This ensures that it holds, in particular, in the local view.
\end{example}

%
%
\costin{Can we give an example ? Maybe on Fig. 2}

\subsection{Local View Argument: Conditions \& Guarantees}
\label{Se:FrameworkConditions}
\yotam{Probably need to use some word other than ``framework'', it really hurts the eye...}
The main theorem underlying our framework bridges the discrepancy between the \emph{local view} of a thread as it performs a sequence of read actions, and the actual \emph{global state} during the traversal.

In the sequel, we fix a sequence of 
read actions $\seqr= r_1,\ldots,r_d$ executed by some thread, and denote the sequence of write actions executed concurrently with $\seqr$ by $\seqw = w_1,\ldots, w_n$.
We denote the global state when $\seqr$ starts its execution by $\hinit$, and the intermediate global states obtained after each prefix of these writes in $\seqw$ by $\opprefix{\hconc}{i} = w_1 \ldots w_i(\hinit)$.

Using the above terminology, our framework devises conditions for showing for a reachability property $\PredSearch{k}{x}$ that if $\PredSearch{k}{x}(\hlocal)$ holds, then there exists $0 \leq i \leq n$ such that $\PredSearch{k}{x}(\opprefix{\hconc}{i})$ holds, which means that $\PastSymb \PredSearch{k}{x}$ holds in the actual global state reached at the end of the traversal. We formalize these conditions below.

\subsubsection{Condition I: Temporal Acyclicity}
The first requirement of our framework concerns the order on the memory locations representing the data structure,
according to which readers perform their traversals. We require that writers maintain this order \emph{acyclic across intermediate states} of the execution.
%
For example, when the order is based on following pointers in the heap, then, if it is possible to reach location $y$ from location $x$ by following a path in which every pointer was present at \emph{some} point in time (not necessarily the same point), then it is not possible to reach $x$ from $y$ in the same manner.
This requirement is needed in order to ensure that the order is robust even from the perspective of a concurrent reading operation, whose local view is obtained from a fusion of fractions of states.
We begin formalizing this requirement with the notion of search order on memory.

\para{Search order}
The acyclicity requirement is based
on a mapping from a state $H$ to a \emph{partial order} that $H$ induces on memory locations, denoted $\leqloc_H$,
that captures the order in which operations read the different memory locations.
Formally, 
$\leqloc_H$ is a \emph{search order}:
\begin{definition}[Search order]
\label[definition]{def:search-order}
$\leqloc_H$ is a \emph{search order} if it satisfies the following conditions:
\begin{enumerate}[(i)]
	\item It is \emph{locally determined}: if $\ell_2$ is an immediate successor of $\ell_1$ in $\leqloc_H$, then for every $H'$ such that $H'(\ell_1) = H(\ell_1)$ it holds that $\ell_1 \leqloc_{H'} \ell_2$.
\TODO{CE: : I would replace “locally-determined” with “value-determined”. I have the impression that this condition is not about locality, but about saying that the order is determined by values stored in locations.}
	\item \label{it:search-order-path} Search paths follow the order: if there is a $k$-search path between $\ell_1$ and $\ell_2$ in $H$, then $\ell_1 \leqloc_H \ell_2$.
	\item \label{it:search-order-read} Readers follow the order: 
	reads in $\seqr$ always read a location further in the order in the current global state. Namely, if $\ell'$ is the last location read, the next read $r$ reads a location $\ell$ from the state $\opprefix{\hconc}{m}$ such that $\ell' \leqloc_{\opprefix{\hconc}{m}} \ell$.
\end{enumerate}
\end{definition}
Note that the locality of the order is helpful for the ability of readers to follow the order: the next location can be known to come forward in the order solely from the last value the thread reads.


\begin{example}
In the example of \Cref{Fi:Running},
the order $\leqloc_H$ is defined by following pointers from parent to children, i.e., all the fields of $x.\textit{left}$ and $x.\textit{right}$ are ordered after the fields of $x$, and the fields of an object are ordered by $x.key < \QXD{x} < \{x.\textit{left}, x.\textit{right}\}$.
%
It is easy to see that this is a search order. Locality follows immediately, and so does the property that search paths follow the order.
The fact that the read-in-order property holds for all the methods in~\Cref{Fi:Running} follows from a very simple syntactic analysis, e.g., in the case of \code{locate(k)}, children are always read after their parents and the field $\textit{key}$ is always accessed before $\textit{left}$ or $\textit{right}$.
\end{example}

\begin{remarknum}
Different search orders may be used for different traversals and different $k$'s when establishing $\PReachXK{x}{k}$ at the end of the traversal. 
In \Cref{def:search-order}, condition (\ref{it:search-order-read}) considers (just) the reads performed by the traversal of interest, and condition (\ref{it:search-order-path}) considers the possible search paths it constructs in the local view (just) for the $k$ of interest.
\end{remarknum}

%
%
%


\para{Accumulated order and acyclicity}
%
The accumulated order captures the order as it may be observed by concurrent traversals across different intermediate states.
 Formally, we define the \emph{accumulated order} w.r.t.\ a sequence of writes $\hat{w}_1,\ldots,\hat{w}_m$, denoted $\leqacculoc_{\hat{w}_1 \ldots \hat{w}_m(\hinit)}$, as the transitive closure of
$\bigcup\limits_{0 \leq s \leq m} \leqloc_{\hat{w}_1 \ldots \hat{w}_s(\hinit)}$.
In our example, the accumulated order consists of all parent-children links created during an execution.
We require:

\begin{definition}[Acyclicity]
\label[definition]{def:cond-acyclicity}
 We say that $\leqloc_H$ satisfies \emph{acyclicity of accumulated order} w.r.t.\ a sequence $\seqw = w_1,\ldots,w_n $ of writes if the accumulated order $\leqacculoc_{w_1 \ldots w_n(\hinit)}$ is a partial order.
\end{definition}

\begin{example}
\label[example]{example:acyclicity-cf}
In our running example, 
acyclicity holds because \code{insert}, \code{remove}, and \code{rotate} modify the pointers from a node only to point to new nodes, or to nodes that have already been reachable from that node. Modifications to other fields have no effect on the order. Note that \code{rotate} does not perform the rotation in place, but allocates a new object.
\sharon{revived following:}
Therefore, the accumulated order, which consists of all parent-children links created during an execution, is acyclic, and hence remains a partial order.
%
\end{example}

\subsubsection{Condition II: Preservation of Search Paths}
The second requirement of our framework is that for every write action
$w$ which happens concurrently with the sequence of reads $\seqr$ and modifies location $\modw{w}$,
if $\modw{w}$ was $k$-reachable (i.e., $\PredSearch{k}{\modw{w}}$ was true) at some point in time after $\seqr$ started and before $w$ occurred,
then it also holds right before $w$ is performed.
We note that this must hold in the presence of all possible interferences, including writes that operate on behalf of other keys (e.g.\ $\code{insert}(k')$).
Formally, we require:


\begin{definition}[Preservation]
\label[definition]{def:cond-preservation}
We say that $\seqw$ \emph{ensures preservation of $k$-reachability by search paths} if for every $1 \leq m \leq n$, 
if for some $0 \leq i < m$, $\opprefix{\hconc}{i} \models \PredSearch{k}{\modw{w_m}}$  then $\opprefix{\hconc}{m-1} \models \PredSearch{k}{\modw{w_m}}$.
\end{definition}
Note that $\opprefix{\hconc}{m-1} \models \PredSearch{k}{\modw{w_m}}$ iff $\opprefix{\hconc}{m} \models \PredSearch{k}{\modw{w_m}}$ since the search path to $\modw{w_m}$ is not affected by $w_m$ (by the basic properties of $\PredSearch{k}{\modw{w_m}}$, see \Cref{Se:ProofPresentSearchPaths}).

\begin{example}
\label[example]{example:preservation-cf}
In our running example, preservation holds because $w_m$ either modifies a location that has never been reachable (such as \cref{Ln:RotateFresh}), in which case preservation holds vacuously, or holds the lock on $x$ when $\QXR[\neg]{x}$ (without modifying its predecessor earlier under this lock).\footnote{
	In \cref{Ln:RotateXRight}, because $x$ is a child of $y$ which is a child of $p$ and $\QXR[\neg]{p}$, it follows that $\QXR[\neg]{x}$ because a node marked with ${\tt rem}$ loses its single parent beforehand.
}
In the latter case preservation holds because every previous write $w'$ retains $\QReachXK{\modw{w_m}}{k}$ unchanged unless it sets the field ${\tt rem}$ of $x$ to $\true$ before releasing the lock on $x$. Therefore, $\QReachXK{\modw{w_m}}{k}$ is retained still when $w_m$ is performed. Preservation follows.
\end{example}

\yotam{Trying to answer Reviewer 1: In particular, I'm not sure I fully understand the second condition, ``preservation''. Why does it need the k-reachability to hold *before* the write is performed? I expected something along the lines of every write preserves all k-paths. Example 7 doesn't really help because the example seems to be satisfying a stronger property.}
We emphasize that the preservation condition only requires that $k$-reachability is retained to modified locations $\ell$ and only at the point of time when the write $w$ to $\ell$ is performed; $k$-reachability \emph{may} be lost at later points in time. In particular, locations whose reachability has been reduced may be accessed, as long as they are not modified after the reachability loss.
For example, consider a rotation as in \Cref{Fi:Rotation}.
The rotation breaks the $k$-reachability of $y$: $\QReachXK{y}{k}$ holds before the rotation but not afterwards.
Indeed, our framework does not establish $\QReachXK{y}{k}$, but infers $\PReachXK{y}{k}$, which does hold.
In this example, the preservation condition requires that the left and right pointers of $y$ are not modified after this rotation is performed.\footnote{Modification of $y.{\tt rem}$ is allowed because this field does not affect search paths (see \Cref{Se:ProofPresentSearchPaths}).}
On the other hand, concurrent traversals may \emph{access} $y$.
In the example, this happens
when (1) the traversal continues beyond $y$ in the search for $k' \neq k$, and when (2) the traversal searches for $k$ and terminates in $y$.

\subsubsection{Local View Arguments' Guarantee}
We are now ready to formalize our main theorem, relating reachability in the local view (\Cref{Se:local-view}) to reachability in the global state, provided that the conditions from \Cref{def:cond-acyclicity,def:cond-preservation} are satisfied.

\begin{theorem}
\label{thm:main-result-instance-reach-short}
If
\begin{inparaenum}[(i)]
	\item $\leqloc_H$ is a search order satisfying the accumulated acyclicity property w.r.t.\ $\seqw$, and
	\item $\seqw$ ensures preservation of $k$-reachability by search paths, 
\end{inparaenum}
then for every $k$ and location $x$, if $\PredSearch{k}{x}(\hlocal)$ holds, then there exists
$0 \leq i \leq n$ s.t.\ $\PredSearch{k}{x}(\opprefix{\hconc}{i})$ holds.
\end{theorem}

In~\refappendix{Se:NecessityConditios} we illustrate how violating these conditions could lead to incorrectness of traversals. \yf{Not accurate, we demonstrate for read in order instead of acyclicity...} \Cref{Se:ProofIdea} discusses the main ideas behind the proof.

\subsection{Proof Idea}
\label{Se:ProofIdea}
We now sketch the correctness proof of \Cref{thm:main-result-instance-reach-short}.
(The full details appear in \refappendix{Se:AbstractFrameworkLong}.)
The theorem transfers $\PredSearch{k}{x}$ from the local view to the global state.
Recall that the local view is a fusion of the fractions of states observed by the thread at different times.
To relate the two, we study the local view from the lens of a fabricated state: a state resulting from a subsequence of the interfering writes, which includes the observed local view.
We exploit the cooperation between the readers and the writers that is guaranteed by the order $\leqloc_H$ (which readers and writers maintain) to construct a fabricated state which is closely related to the global state, in the sense that it \emph{simulates} the global state (\Cref{def:uniform-change-abstract-short}); simulation depends both on the acyclicity requirement and on the preservation requirement (\Cref{lemma:simulation-from-conditions}).
Deducing the existence of a search path in an intermediate global state out of its existence in the local view is a corollary of this connection (\Cref{thm:abstract-framework}).

\para{Fabricated state}
The fabricated state provides a means of analyzing the local view and its relation to the global (true) state.
A \emph{fabricated state} is a state consistent with the local view (i.e.\ it agrees with the value of every location present in the local view) that is constructed by a subsequence $\seqwf = w_{i_1},\ldots, w_{i_k}$ of the writes $\seqw$.
One possible choice for $\seqwf$ is the subsequence of writes whose effect was observed by $\seqr$ (i.e.\ $\seqr$ read-from).
For relating the local view to the global state, which is constructed from the entire $\seqw$, it is beneficiary to include in $\seqwf$ additional writes except for those directly observed by $\seqr$.
In what follows, we choose the subsequence $\seqwf$ so that the fabricated state satisfies a consistency property of \emph{forward-agreement} with the global state. This means that although not all writes are included in $\seqwf$ (as the thread misses some),
the writes that are included 
have the same picture of the ``continuation'' of the data structure as it really was in the global state.

\para{Construction of fabricated state based on order}
Our construction of the fabricated state includes in $\seqwf$ 
all the writes that occurred \emph{backward} in time and wrote to locations \emph{forward} in the order than the current location read, for every location read.
(In particular, it includes all the writes that $\seqr$ reads from directly).
%
Formally, let $\modw{w}$ denote the location modified by write $w$.
Then for every read $r$ in $\seqr$ that reads location $\ell_r$ from global state $\opprefix{\hconc}{m}$, we include in $\seqwf$ all the writes $\{w_j \mid j \leq m \land \ell_r \leqacculoc_{w_1 \ldots w_m(\hinit)} \modw{w_j}\}$ (ordered as in $\seqw$).
We use the notation $\opprefix{\hobserved}{j} = w_{i_1}\ldots w_{i_j}(\hinit)$ for intermediate fabricated states.
This choice of $\seqwf$ ensures \emph{forward-agreement} between the fabricated state and the global state: every write $w_{i_j}$ in $\seqwf$, the states on which it is applied, $\opprefix{\hconc}{i_j-1}$ and $\opprefix{\hobserved}{j-1}$ agree on all locations $\ell$ such that $\modw{w_{i_j}} \leqloc_{\opprefix{\hobserved}{j-1}} \ell$.

In what follows, we fix the fabricated state to be the state resulting at the end of this particular choice of $\seqwf$.
It satisfies forward-agreement by construction, and is an extension of the local view,
relying on the \emph{acyclicity} requirement.

\sharon{I can live with this version, but note that forward agreement repeats 3 times... twice with an explanation. YF: changed one to be more formal}


\para{Simulation}
As we show next, the construction of $\seqwf$ ensures that the effect of every write in $\seqwf$ on $\PredSearch{k}{x}$ is guaranteed to concur with its effect on the real state with respect to changing $\PredSearch{k}{x}$ from false to true. We refer to this property as \emph{simulation}.


\begin{definition}[Simulation]
\label[definition]{def:uniform-change-abstract-short}
For a predicate $\Pred$, we say that the subsequence of writes $w_{i_1} \ldots w_{i_k}$ \emph{$\Pred$-simulates} the sequence $w_1 \ldots w_n$ 
if for every $1 \leq j \leq k$,
if $\neg\Pred(\opprefix{\hobserved}{j-1})$ but $\Pred(w_{i_j}(\opprefix{\hobserved}{j-1}))$, then
$\neg\Pred(\opprefix{\hconc}{i_{j}-1}) \implies \Pred(w_{i_j}(\opprefix{\hconc}{i_{j}-1}))$.
%
\end{definition}

Simulation implies that the write $w_{i_j}$ in $\seqwf$ that changed $\PredSearch{k}{x}$ to true on the local view, would also change it on the corresponding global state $\opprefix{\hconc}{i_j}$ (unless it was already true in $\opprefix{\hconc}{i_j-1}$). This provides us with the desired global state where $\PredSearch{k}{x}$ holds.
Using also the fact that $\PredSearch{k}{x}$ is upward-absolute~\cite{absoluteness} (namely, preserved under extensions of the state), we obtain:
\begin{lemma} \label[lemma]{thm:abstract-framework}
Let $\seqwf$ be the subsequence of $\seqw = w_{1},\ldots, w_{n}$ defined above.
If $\PredSearch{k}{x}(\hlocal)$ holds and $\seqwf$ $\PredSearch{k}{x}$-simulates 
$\seqw$, then there exists some $0 \leq i \leq n$ s.t.\ $\PredSearch{k}{x}(\opprefix{\hconc}{i})$.
\end{lemma}

Finally, we show that the fabricated state satisfies the simulation property.
Owing to the specific construction of $\seqwf$, the proof needs to relate the effect of writes on states which have a rather strong similarity:
they agree on the contents of locations which come forward of the modified location. Preservation complements this by guaranteeing the existence of a path to the modified location:
\begin{lemma}\label[lemma]{lemma:simulation-from-conditions}
If $\seqw$ satisfies preservation of $\PredSearch{k}{\modw{w}}$ for all $w$, then $\seqwf$ $\PredSearch{k}{x}$-simulates 
$\seqw$ for all $x$.
\end{lemma}
To prove the lemma, we show that preservation, together with forward agreement, implies the simulation property, which in turn implies that
$\PredSearch{k}{x}(\opprefix{\hobserved}{j-1}) \implies \exists 0 \leq i \leq i_{j-1} \ \PredSearch{k}{x}(\opprefix{\hconc}{i})$  (see \Cref{thm:abstract-framework}).
%
%
To show simulation, consider a write $w_{i_j}$ 
that creates a $k$-search path $\zeta$ to $x$ in $\opprefix{\hobserved}{j}$. We construct such a path in the corresponding global state. The idea is to
divide $\zeta$ to two parts: the prefix until $\modw{w_{i_j}}$, and the rest of the path.
Relying on forward agreement, the latter is exactly the same in the corresponding global state, and preservation lets us prove that there is also an appropriate prefix: necessarily there has been a $k$-search path to $\modw{w_{i_j}}$ in the fabricated state \emph{before} $w_{i_j}$, so by induction, exploiting the fact that simulation up to $j-1$ implies that $\PredSearch{k}{x}(\opprefix{\hobserved}{j-1}) \implies \exists 0 \leq i \leq i_{j-1}. \ \PredSearch{k}{x}(\opprefix{\hconc}{i})$,
there has been a $k$-search path to $\modw{w_{i_j}}$ in some intermediate global state that occurred earlier than the time of $w_{i_j}$.
Since $w_{i_j}$ writes to $\modw{w_{i_j}}$, the preservation property ensures that there is a $k$-search path to $\modw{w_{i_j}}$ in the global state also at the time of the write $w_{i_j}$, and the claim follows.



\section{Putting It All Together: Proving Linearizability Using Local Views}
\label{Se:PuttingTogetherCF}

Recall that our overarching objective in developing the local view argument (\Cref{Se:LocalViewArguments}) is to prove the correctness of assertions used in linearizability proofs (e.g., in~\Cref{Se:ProofLin}).
We now summarize the steps in the proof of the assertions.
Overall, it is composed of the following steps:
\begin{enumerate}
	\item \label{it:prove-assert:seq-reasoning} Establishing properties of traversals on the local view using sequential reasoning,
	\item \label{it:prove-assert:conditions} Establishing the acyclicity and preservation conditions by simple concurrent reasoning, and
	\item \label{it:prove-assert:rest} Proving the assertions when relying on local view arguments, augmented with some concurrent reasoning.
\end{enumerate}

For the running example, step \ref{it:prove-assert:seq-reasoning} is presented in \Cref{example:seq-reason-cf}, and step \ref{it:prove-assert:conditions} consists of \Cref{example:acyclicity-cf,example:preservation-cf} (see \refappendix{Se:FormalRunningConditions} for a full formal treatment). Step \ref{it:prove-assert:rest} concludes the proof as discussed in \Cref{Se:RunningProvingAssertions}.


\begin{remarknum}\label{remark:assertions-based-reasoning}
While the local view argument, relying in particular on step \ref{it:prove-assert:conditions}, was developed to simplify the proofs of the assertions in \ref{it:prove-assert:rest},
this goes also in the other direction. Namely,
the concurrent reasoning required for proving the conditions of the framework (e.g., preservation) \yotam{Omitted (I guess it refers to rely-guarantee, but its too subtle, and it's standard, no?): as well as the assertions}
can be greatly simplified by relying on the correctness of the assertions (as they constrain possible interfering writes).
Indeed, the proofs may mutually rely on each other.
This is justified by a proof by induction: we prove that the current write satisfies the condition in the assertion, assuming that all previous writes did.
This is also allowed in proofs of the conditions in \Cref{Se:FrameworkConditions}, because they refer to the effect of \emph{interfering} writes, that are known to conform to their respective assertions from the induction hypothesis. Hence, carrying these proofs together avoids circular reasoning and ensures validity of the proof.
\end{remarknum}

\section{Additional Case Studies} \label{Se:CaseStudies}\label{Se:Using:Lazy}

\subsection{Lazy and Optimistic Lists}

We successfully applied our framework to prove the linearizability of sorted-list-based concurrent set implementations with unsynchronized reads.
Our framework is capable of verifying various versions of the algorithm in which \code{insert} and \code{delete} validate that the nodes they locked are reachable using a boolean field, as done in the lazy list algorithm~\cite{HellerHLMMS05},  or by rescanning the list, as done in the optimistic list algorithm~\cite[Chap\,9.8]{TAOMPP}. Our framework is also applicable for verifying implementations of the lazy list algorithm in which the logical deletion and the physical removal are done by the same operation or by different ones.
We give a taste of these proofs here. 

\Cref{Fi:Lazy-List} shows an annotated pseudo-code of the lazy list algorithm.
Every operation starts with a call to \code{locate(k)}, which performs a standard search in a sorted list---without acquiring any locks---to locate the node with the target key $k$.
This method returns the last link it traverses, $(x,y)$.
\Cref{Fi:Lazy-List} includes two variants of \code{contains(}$k$\code{)}:
In one variant, it returns \code{true} only if it finds a node with key $k$ that is not logically deleted (\cref{LnLazy:ContainsReturnTorT}), while in the second variant it returns \code{true} even if that node is logically deleted (the commented \code{return} at \cref{LnLazy:ContainsReturnTorF}).
Interestingly, the same annotations allow to verify both variants, and the proof differs only in the abstraction function mapping states of the list to abstract sets.
Modifications of a node in the list are synchronized with the node's lock.
An \code{insert(k)} operation calls \code{locate}, and then links a new node to the list if $k$ was not found.
\code{delete(k)} logically deletes $y$ (after validating that $y$ remained linked
to the list after its lock was acquired), and then physically removes it.

As in \Cref{Se:Motivating}, the assertions contain predicates 
of the form
$\QReachXK{x}{k}$, which means that $x$ resides on a \emph{valid search path} for
key $k$ that starts at $\rootobj$; 
the formal definition of a 
search path in the lazy list appears below. Note that $\QReachXK{\NULL}{k}$ indicates that $k$ is not in the list.
$$
\begin{array}{l}
\QReachXYK{o_r}{o_x}{k}
\eqdef
\exists o_0,\ldots,o_m.\
o_0=o_r \land o_m=o_x \land
\forall i=1..m.\ o_{i-1}.\ikey < k \land o_{i-1}.\inext = o_i
\end{array}
$$

We  prove the linearizability of the algorithm  using an \emph{abstraction function}.
One  abstraction function we may use maps $H$ to the set of keys of the nodes that are on a valid search path for their key and
are not logically deleted in $H$:
$$
\repfunc^{\mathit{logical}}(H) = \{ k \in \mathbb{N} \mid H \models \exists x.\,\QReachXK{x}{k} \land x.\ikey=k \land \neg\QXM{x}\}\,.
$$





\lstset{language={program},style=lnumbers,firstnumber=last}

\begin{figure}[t]
\centering
\begin{tabular}{p{7cm}p{6.7cm}p{6.7cm}}
\begin{lstlisting}
type N 
  int key 
  N next 
  bool mark

N root$\leftarrow$new N($-\infty$);

N$\times$N locate(int k)
  x,y$\leftarrow$root
  while (y$\neq$null $\land$ y$.$key$<$k)
    x$\leftarrow$y 
    y$\leftarrow$x$.$next
  $
  \{\Past{\QReachXK{x}{k} \land x.\inext=y }\}   \label{LnLazy:RetLocate}
  $ 
  $
  \{\QXK[<]{x}{k} \land (y \neq \NULL \implies  \QXK[\geq]{y}{k})\}  \label{LnLazy:RetLocateKeys} 
  $ 
  return (x,y)

bool insert(int k)  
  (x,y)$\leftarrow$locate(k) 
  if (y$\neq$null $\land$ y$.$key$=$k)
    $
      \{\PReachXK{y}{k} \land \QXK{y}{k}\}
    $   
    return false
  lock(x) 
  lock(y)
  if (x$.$mark $\lor$ x$.$next$\neq$y)
    restart
  $
     \{\neg \QXM{x} \land x.\inext=y \}    \label{LnLazy:InsertValidate}
  $    
  z$\leftarrow$new N(k)
  $
  \{y \neq \NULL \implies  \QXKR[>]{y}{k}\}   
  $ 
  z$.$next$\leftarrow$y $\label{LnLazy:InsertPrepareBypass}$
 $
  \begin{array}{l}
     \{\QReachXK{x}{k} \land x.\inext=y  \land \QXK[<]{x}{k}   \land {}\\
     \qquad z.\inext=y \land \neg \QXM{z} \land (y \neq \NULL \implies  \QXKR[>]{y}{k})  \}
  \end{array}
  \label{LnLazy:InsertLink}
  $    
  x$.$next$\leftarrow$z  $\label{LnLazy:InsertSetNext}$
  return true
\end{lstlisting}

&


\begin{lstlisting}
bool contains(int k)  
  (_,y)$\leftarrow$locate(k) 
  if (y$=$null)
    $
      \{\PReachXK{\NULL}{k}\}     \label{LnLazy:ContainsNull}
    $   
    return false
  if (y.key$\neq$k)
    $
     \{\Past{\QReachXK{x}{k} \land x.\inext=y} \land \QXKR[<]{x}{k} \land \QXK[>]{y}{k} \}   \label{LnLazy:ContainsNoKey}
    $ 
    return false 
  if ($\neg$y.mark)
    $
     \{\QReachXK{y}{k} \land \QXK{y}{k} \land \neg\QXM{y} \}   \label{LnLazy:ContainsKeyNotMark}
    $ 
    return true $\label{LnLazy:ContainsReturnTorT}$
  $
    \{\PReachXK{y}{k} \land \QXK{y}{k} \land \QXM{y}  \}   \label{LnLazy:ContainsKeyMark}
  $ 
  return false // return true $\label{LnLazy:ContainsReturnTorF}$

bool delete(int k) 
  (x,y)$\leftarrow$locate(k) 
  if (y$=$null)
    $
      \{\PReachXK{\NULL}{k}\}
    $   
    return false
  if (y.key$\neq$k)
    $
     \{\Past{\QReachXK{x}{k} \land x.\inext=y} \land \QXK[<]{x}{k} \land  \QXK[>]{y}{k} \}   
    $ 
    return false
  $
  \{\QXK{y}{k}\}
  $   
  lock(x) 
  lock(y)
  if (x$.$mark $\lor$ y$.$mark $\lor$ x$.$next$\neq$y)
    restart 
  $
   %  \{\neg \QXM{x} \land x.\inext=y \land \QXK{y}{k}\}    \label{LnLazy:DeleteValidate}
     \{\QReachXK{x}{k} \land x.\inext=y    \land \QXK{y}{k} \land \neg\QXM{x} \land 
     \neg\QXM{y} \}   \label{LnLazy:DeleteLogical} \label{LnLazy:DeleteValidate}
  $    
  y$.$mark$\leftarrow$true                   $\label{LnLazy:DeleteSetMark}$
  $
     \{\QReachXK{x}{k} \land x.\inext=y    \land \QXK{y}{k} \land \neg\QXM{x} \land
      \QXM{y} \}   
  $ 
  x$.$next$\leftarrow$y$.$next                $\label{LnLazy:DeleteRemove}$
  return true                               $\label{LnLazy:DeleteRetTrue}$
\end{lstlisting}
\end{tabular} 
\vspace{-7mm}
\caption{\label{Fi:Lazy-List}
Lazy List~\cite{HellerHLMMS05}. The code is annotated with assertions written inside curly braces.
For brevity, \textbf{unlock} operations are omitted; a procedure releases all the locks it acquired when it terminates or \textbf{restart}s.
} 
\vspace{-3mm}
\end{figure}
%
%
%
%
Another possibility is to define the abstract set to be the keys of all the reachable nodes:
$$
\repfunc^{\mathit{physical}}(H) = \{ k \in \mathbb{N} \mid H \models \exists x.\,\QReachXK{x}{k} \land x.\ikey=k\}\,.
$$
We note that $\repfunc^{\mathit{logical}}(H)$ can be used to verify the code of \code{contains} as written,
while $\repfunc^{\mathit{physical}}(H)$ allows to change   the algorithm to return \code{true} in \cref{LnLazy:ContainsReturnTorF}.
In both cases, the proof of linearizability is carried out using the same assertions currently annotating the code.
In the rest of this section, we discuss the verification of the code in \Cref{Fi:Lazy-List} as written, and thus use  $\repfunc(H)=\repfunc^{\mathit{logical}}$ as the abstraction function.
The assertions almost immediately imply that for every operation invocation $op$, there exists a state $H$ during
$op$'s execution for which the abstract state $\repfunc(H)$ agrees with $op$'s return value, and so $op$ can
be linearized at $H$;
\sharonx{``Method'' vs.\ operation}%
we need only make the following observations.  First, \code{contains()} and a failed \code{delete()}
or \code{insert()} do not modify the memory, and so can be linearized at the point in time in which the assertions
before their \code{return} statements hold.  Second, in the state $H$ in which a successful \code{delete(k)}
(respectively, \code{insert(k)}) performs a write, the assertions on line~\ref{LnLazy:DeleteLogical}
(respectively, \cref{LnLazy:InsertLink}) imply that $k \in \repfunc(H)$
(respectively, $k \not\in \repfunc(H)$).  Therefore, these writes change the abstract set, making it agree
with the operation's return value of $\true$.
Finally, it only remains to verify that the physical removal performed by \code{delete(}$k$\code{)}  in state $H$ does not modify $\repfunc(H)$.
Indeed, as an operation modifies a field of node $v$ only when it has $v$ locked,
it is easy to see that for any node $x$ and key $k$, if $\QReachXK{x}{k}$ held before the write, then it also
holds afterwards with the exception of the removed node $y$. However, \code{delete(}$k$\code{)} removes a deleted node, and thus does not change $\repfunc(H)$.

The proof of the assertions in \Cref{Fi:Lazy-List} utilizes a local view argument for the $\PastSymb$ assertion in~\cref{LnLazy:RetLocate} for the predicate $\QReachXK{x}{k} \land x.\inext=y$, using the extension with a single field discussed in~\Cref{rem:ReachWithField}. 
The conditions of the local view argument are easy to prove:
The acyclicity requirement is evident, as writes modify the pointers from a node only to point to new nodes, or to nodes that have already been reachable from that node.
Preservation holds because a write either (i) marks a node, which does not affect the search paths; (ii) modifies a location that has never been reachable (such as
\cref{LnLazy:InsertPrepareBypass}), 
in which case preservation holds vacuously;
(iii) removes a marked node $y$ (\cref{LnLazy:DeleteRemove}) which removes all the search paths that go through it. However, as $y$ is marked, its fields are not going to be modified later on, and thus $y$ cannot be the cause of violating preservation. Furthermore, all search paths that reach $y$'s successor before the removal are retained and merely get shorter; or
(iv)  adds a reachable node $z$ in between two reachable nodes $x$ and $y$ 
(\Cref{LnLazy:InsertSetNext}). However, as $z$'s key is smaller than $y$'s, the insertion preserves any search paths which goes \emph{through} $y$'s \code{next} pointer.


As for the rest of the assertions, when \code{insert} and \code{delete} lock $x$ and see that it is not marked, the $\QReachXK{x}{k}$ property follows from the $\PReachXK{x}{k}$ deduced above by a local view argument using the same invariant in preservation above.\footnote{As in \Cref{Se:CaseLockFreeList}, these assertions could also be deduced directly from a slightly stronger invariant that unmarked nodes are reachable and that the list is sorted. This is not the case in the optimistic list of~\cite[Chap\,9.8]{TAOMPP} which rescans instead of using a marked bit. In both cases \code{contains} requires a local view argument.}
The remainder assertions are attributed to reading a location under the protection of a lock, e.g.\ $\neg\QXM{x}$ in~\cref{LnLazy:InsertValidate}.

\subsection{Lock-free List and Skip-List}
\label{Se:CaseLockFreeList}\label{Se:CaseLockFreeSkipList}
We used our framework to prove the linearizability  a sorted lock-free list-based concurrent set algorithm~\cite[Chapter 9.8]{TAOMPP} and of a lock-free skip-list-based concurrent set algorithm~\cite[Chapter 14.4]{TAOMPP}.
In these proofs we use local view arguments to prove the concurrent traversals of the \code{contains} method, which is the most difficult part of the proofs:
\code{add} and \code{remove} use the internal \code{find} which traverses the list
and also prunes out marked nodes, and thus their correctness follows easily from an invariant ensuring the reachability of unmarked nodes.
\yotam{\textbf{this is true also for the lazy list (but not the optimistic), no?}}
\iflong
The proofs appear in \Cref{Se:LockFreeList,Se:LockFreeSkipList}.
\else
The proofs appear in~\cite{extendedVersion}.
\fi









\section{Related Work}\label{Se:Related}

Verifying linearizability of concurrent data structures 
has been studied extensively.
Some techniques, 
e.g., ~\cite{conf/tacas/AbdullaHHJR13,conf/cav/AmitRRSY07,conf/cav/DragoiGH13,conf/ppopp/VafeiadisHHS06,conf/vmcai/Vafeiadis09}, apply to a restricted set of algorithms where the linearization point of every invocation is \emph{fixed} to a particular statement in the code. While these works provide more automation, they are not able to deal with the algorithms considered in our paper where for instance, the linearization point of \code{contains(k)} invocations is not fixed. Generic reductions of verifying linearizability to checking a set of assertions in the code have been defined in~\cite{conf/esop/BouajjaniEEH13,DBLP:conf/icalp/BouajjaniEEH15,DBLP:conf/cav/BouajjaniEEM17,DBLP:conf/pldi/LiangF13,conf/concur/HenzingerSV13,conf/cav/Vafeiadis10,DBLP:conf/cav/ZhuPJ15}. These works apply to algorithms with non-fixed linearization points, but they do not provide a systematic methodology for proving the assertions, which is the main focus of our paper.

Verifying linearizability has also been addressed in the context of defining program logics for \emph{compositional} reasoning about concurrent programs. In this context, the goal is to define a proof methodology that allows   composing proofs of program's components to get a proof for the entire program, which can also be reused in every valid context of using that program. Improving on the classical Owicki-Gries~\cite{DBLP:journals/cacm/OwickiG76} and Rely-Guarantee~\cite{DBLP:conf/ifip/Jones83} logics, various extensions of Concurrent Separation Logic~\cite{DBLP:conf/popl/BornatCOP05,DBLP:conf/concur/Brookes04,DBLP:conf/concur/OHearn04,DBLP:conf/popl/ParkinsonBO07} have been proposed in order to reason compositionally about different instances of fine-grained concurrency, e.g.~\cite{DBLP:conf/popl/JungSSSTBD15,DBLP:conf/popl/Ley-WildN13,DBLP:conf/ecoop/PintoDG14,DBLP:conf/esop/SergeyNB15,DBLP:conf/icfp/TuronDB13,phd/Vafeiadis08}. 
However, they focus on the reusability of a proof of a component in a larger context (when composed with other components) while our work focuses on simplifying the proof goals that guarantee linearizability. The concurrent reasoning needed for our framework could be carried out using one of these logics.


The proof of  linearizability of the lazy-list algorithm given in~\cite{PODC10:Hindsight} is based on establishing the conditions required by the 
\emph{hindsight lemma}~\cite[Lemma~5.2]{PODC10:Hindsight}. The lemma states that every link traversed during an unsynchronized traversal was indeed reachable at some point in time between the beginning of the traversal and the moment the link was crossed. This enables verifying the correctness of the \code{contains} method using, effectively, sequential reasoning. The hindsight lemma is a specific instance of the extension discussed in \Cref{rem:ReachWithField}, and its assumptions narrows its application to concurrent set algorithms implemented using   sorted singly-linked lists.
In contrast, we present a fundamental technique which is based on far more generic properties which  is applicable to list and tree-based data structures alike.


The proof methodology for proving linearizability of~\cite{DBLP:conf/wdag/Lev-AriCK15}  relies on properties
of the data structure in sequential executions.
The methodology assumes the existence of {\em base points}, which are points in time during the concurrent
execution of a search in which some predicate holds over the shared state.  For instance, when applying the
methodology to the lazy list, they prove the existence of base points using prior
techniques~\cite{PODC10:Hindsight,LazyListSafety} that employ tricky concurrent reasoning.
Our work is thus complementary to theirs: our proof argument is meant to replace the latter
kind of reasoning, and can thus simplify proofs of the existence of base points.

The {\em Edgeset} framework of Shasha and Goodman~\cite{DBLP:journals/tods/ShashaG88},
which has recently been formalized using concurrent separation logic~\cite{DBLP:journals/pacmpl/KrishnaSW18},
provides conditions for the linearizability of concurrent search data structures.
It relies on a precondition that for any operation on key $k$, $\QReachXK{x}{k}$ holds
when the operation looks for, inserts, or deletes $k$ at $x$.
However, the optimistic data structures that we consider often do not satisfy this precondition,
making the Edgeset framework inapplicable.  (\Cref{example:preservation-cf} describes
how this precondition does not hold in our search tree example, and a similar issue
exists in the lazy-list.)
Moreover, the Edgeset precondition implies that the linearization point of an operation occurs at one of its
own atomic steps.  Our framework does not have this requirement.
Shasha and Goodman also describe three algorithm templates and prove, using concurrent
reasoning, that these templates satisfy the preconditions of the Edgeset framework.
In contrast, our argument uses sequential reasoning for traversals, and our concurrent proofs consider only
the effects of interleaving writes---not both reads and writes.

\section{Conclusions and Future Work}\label{Se:Conc}

This paper presents a novel approach for constructing 
linearizability proofs of
concurrent search data structures. We present a general proof argument that is applicable to many existing algorithms, uncovering fundamental structure---the acyclicity and preservation conditions---shared by them.
%
We have instantiated our framework for a self-balancing binary search tree, lists with lazy~\cite{HellerHLMMS05} or non-blocking~\cite{TAOMPP} synchronization, and a lock-free skip list. 
To the best of our knowledge, our work is the first to prove linearizability of a self-balancing binary search tree using a unified proof argument.



An important direction for future work is the mechanism of backtracking.
Some algorithms, including the original CF tree~\cite{EuroPar13:CFTree,PFL16:CFTree}, backtrack instead of restarting when their optimistic validation fails. In the CF tree, backtracking is implemented by directing pointers from child to parent, breaking our acyclicity requirement.
A similar situation arises in the in-place rotation of~\cite{DBLP:conf/ppopp/BronsonCCO10}.
Handling these scenarios in our proof argument is an interesting direction for future work.

An additional direction to explore is validations performed during traversals. For example, the SnapTree algorithm~\cite{DBLP:conf/ppopp/BronsonCCO10} performs in-place rotations which violate preservation. The algorithm overcomes this by performing hand-over-hand validation during a lock-free traversal.
This validation, consisting of re-reading previous locations and ensuring version numbers have not changed, does not fit our approach of reasoning sequentially about traversals.

The preservation of reachability to location of modification arises naturally out of the correctness of traversals in modifying operations, ensuring that the conclusion of the traversal---the existence of a path---holds not only in some point in the past, but also holds at the time of the modification.
We 
show that, surprisingly, preservation, when it is combined with the order, suffices to reason about the traversal by a local view argument. We 
base the correctness of read-only operations on the same predicates, and so rely on the same property. It would be interesting to explore different criteria which ensure the simulation of the fabricated state constructed based on the accumulated order.

Finding ways to extend the framework in these directions is an interesting open problem. This notwithstanding, we believe that our framework captures important principles underlying modern highly concurrent data structures that could prove useful both for structuring linearizability proofs and elucidating the correctness principles behind new concurrent data structures.




\bigskip 

\subparagraph{Acknowledgments.}
  This publication is part of a project that has received funding from the European Research Council (ERC) under the European Union's Horizon 2020 research and innovation programme (grant agreements No [759102-SVIS] and [678177]).
  The research was partially supported by Len Blavatnik and the Blavatnik Family foundation, the Blavatnik Interdisciplinary Cyber Research Center, Tel Aviv University, 
  the United States-Israel Binational Science Foundation (BSF) grants No.\ 2016260 and 2012259,
  and the Israeli Science Foundation (ISF) grant No.\ 2005/17.
  We thank the anonymous reviewers whose comments helped improve the paper.


 
\clearpage

\bibliography{biblio,violin-short}

\begin{thebibliography}{10}

\bibitem{conf/tacas/AbdullaHHJR13}
Parosh~Aziz Abdulla, Fr{\'e}d{\'e}ric Haziza, Luk{\'a}s Hol\'{\i}k, Bengt
  Jonsson, and Ahmed Rezine.
\newblock An integrated specification and verification technique for highly
  concurrent data structures.
\newblock In {\em TACAS}, pages 324--338, 2013.

\bibitem{conf/cav/AmitRRSY07}
Daphna Amit, Noam Rinetzky, Thomas~W. Reps, Mooly Sagiv, and Eran Yahav.
\newblock Comparison under abstraction for verifying linearizability.
\newblock In {\em CAV '07}, volume 4590 of {\em LNCS}, pages 477--490, 2007.

\bibitem{Arbel:2014}
Maya Arbel and Hagit Attiya.
\newblock {Concurrent Updates with RCU: Search Tree As an Example}.
\newblock In {\em Proceedings of the 2014 ACM Symposium on Principles of
  Distributed Computing}, PODC '14, pages 196--205, New York, NY, USA, 2014.
  ACM.
\newblock URL: \url{http://doi.acm.org/10.1145/2611462.2611471}, \href
  {http://dx.doi.org/10.1145/2611462.2611471}
  {\path{doi:10.1145/2611462.2611471}}.

\bibitem{DBLP:conf/popl/BornatCOP05}
Richard Bornat, Cristiano Calcagno, Peter~W. O'Hearn, and Matthew~J. Parkinson.
\newblock Permission accounting in separation logic.
\newblock In Jens Palsberg and Mart{\'{\i}}n Abadi, editors, {\em Proceedings
  of the 32nd {ACM} {SIGPLAN-SIGACT} Symposium on Principles of Programming
  Languages, {POPL} 2005, Long Beach, California, USA, January 12-14, 2005},
  pages 259--270. {ACM}, 2005.
\newblock URL: \url{http://doi.acm.org/10.1145/1040305.1040327}, \href
  {http://dx.doi.org/10.1145/1040305.1040327}
  {\path{doi:10.1145/1040305.1040327}}.

\bibitem{conf/esop/BouajjaniEEH13}
Ahmed Bouajjani, Michael Emmi, Constantin Enea, and Jad Hamza.
\newblock Verifying concurrent programs against sequential specifications.
\newblock In {\em ESOP '13}, volume 7792 of {\em LNCS}, pages 290--309.
  Springer, 2013.

\bibitem{DBLP:conf/icalp/BouajjaniEEH15}
Ahmed Bouajjani, Michael Emmi, Constantin Enea, and Jad Hamza.
\newblock On reducing linearizability to state reachability.
\newblock In {\em Automata, Languages, and Programming - 42nd International
  Colloquium, {ICALP} 2015, Kyoto, Japan, July 6-10, 2015, Proceedings, Part
  {II}}, pages 95--107, 2015.

\bibitem{DBLP:conf/cav/BouajjaniEEM17}
Ahmed Bouajjani, Michael Emmi, Constantin Enea, and Suha~Orhun Mutluergil.
\newblock Proving linearizability using forward simulations.
\newblock In Rupak Majumdar and Viktor Kuncak, editors, {\em Computer Aided
  Verification - 29th International Conference, {CAV} 2017, Heidelberg,
  Germany, July 24-28, 2017, Proceedings, Part {II}}, volume 10427 of {\em
  Lecture Notes in Computer Science}, pages 542--563. Springer, 2017.
\newblock URL: \url{https://doi.org/10.1007/978-3-319-63390-9_28}, \href
  {http://dx.doi.org/10.1007/978-3-319-63390-9_28}
  {\path{doi:10.1007/978-3-319-63390-9_28}}.

\bibitem{DBLP:conf/ppopp/BronsonCCO10}
Nathan~Grasso Bronson, Jared Casper, Hassan Chafi, and Kunle Olukotun.
\newblock A practical concurrent binary search tree.
\newblock In {\em Proceedings of the 15th {ACM} {SIGPLAN} Symposium on
  Principles and Practice of Parallel Programming, {PPOPP} 2010, Bangalore,
  India, January 9-14, 2010}, pages 257--268, 2010.

\bibitem{DBLP:conf/concur/Brookes04}
Stephen~D. Brookes.
\newblock A semantics for concurrent separation logic.
\newblock In Gardner and Yoshida \cite{DBLP:conf/concur/2004}, pages 16--34.
\newblock URL: \url{https://doi.org/10.1007/978-3-540-28644-8_2}, \href
  {http://dx.doi.org/10.1007/978-3-540-28644-8_2}
  {\path{doi:10.1007/978-3-540-28644-8_2}}.

\bibitem{Brown:2014:GTN}
Trevor Brown, Faith Ellen, and Eric Ruppert.
\newblock A general technique for non-blocking trees.
\newblock In {\em PPoPP}, 2014.

\bibitem{Clements:2012}
Austin~T. Clements, M.~Frans Kaashoek, and Nickolai Zeldovich.
\newblock Scalable address spaces using {RCU} balanced trees.
\newblock In {\em ASPLOS}, 2012.

\bibitem{EuroPar13:CFTree}
Tyler Crain, Vincent Gramoli, and Michel Raynal.
\newblock A contention-friendly binary search tree.
\newblock In Felix Wolf, Bernd Mohr, and Dieter an~Mey, editors, {\em Euro-Par
  2013 Parallel Processing}, pages 229--240, Berlin, Heidelberg, 2013. Springer
  Berlin Heidelberg.

\bibitem{NoHotSpotSkipList}
Tyler Crain, Vincent Gramoli, and Michel Raynal.
\newblock {No Hot Spot Non-blocking Skip List}.
\newblock In {\em ICDCS}, 2013.

\bibitem{PFL16:CFTree}
Tyler Crain, Vincent Gramoli, and Michel Raynal.
\newblock A fast contention-friendly binary search tree.
\newblock {\em Parallel Processing Letters}, 26(03):1650015, 2016.
\newblock URL:
  \url{http://www.worldscientific.com/doi/abs/10.1142/S0129626416500158}, \href
  {http://arxiv.org/abs/http://www.worldscientific.com/doi/pdf/10.1142/S0129626416500158}
  {\path{arXiv:http://www.worldscientific.com/doi/pdf/10.1142/S0129626416500158}},
  \href {http://dx.doi.org/10.1142/S0129626416500158}
  {\path{doi:10.1142/S0129626416500158}}.

\bibitem{DBLP:conf/ecoop/PintoDG14}
Pedro da~Rocha~Pinto, Thomas Dinsdale{-}Young, and Philippa Gardner.
\newblock Tada: {A} logic for time and data abstraction.
\newblock In Richard~E. Jones, editor, {\em {ECOOP} 2014 - Object-Oriented
  Programming - 28th European Conference, Uppsala, Sweden, July 28 - August 1,
  2014. Proceedings}, volume 8586 of {\em Lecture Notes in Computer Science},
  pages 207--231. Springer, 2014.
\newblock URL: \url{https://doi.org/10.1007/978-3-662-44202-9_9}, \href
  {http://dx.doi.org/10.1007/978-3-662-44202-9_9}
  {\path{doi:10.1007/978-3-662-44202-9_9}}.

\bibitem{David:2015:ACS}
Tudor David, Rachid Guerraoui, and Vasileios Trigonakis.
\newblock {Asynchronized Concurrency: The Secret to Scaling Concurrent Search
  Data Structures}.
\newblock In {\em ASPLOS}, 2015.

\bibitem{Drachsler:2014}
Dana Drachsler, Martin Vechev, and Eran Yahav.
\newblock {Practical Concurrent Binary Search Trees via Logical Ordering}.
\newblock In {\em PPoPP}, 2014.

\bibitem{conf/cav/DragoiGH13}
Cezara Dragoi, Ashutosh Gupta, and Thomas~A. Henzinger.
\newblock Automatic linearizability proofs of concurrent objects with
  cooperating updates.
\newblock In {\em CAV '13}, volume 8044 of {\em LNCS}, pages 174--190.
  Springer.

\bibitem{Ellen:2010}
Faith Ellen, Panagiota Fatourou, Eric Ruppert, and Franck van Breugel.
\newblock {Non-blocking Binary Search Trees}.
\newblock In {\em PODC}, 2010.

\bibitem{extendedVersion}
Yotam M.~Y. Feldman, Constantin Enea, Adam Morrison, Noam Rinetzky, and Sharon
  Shoham.
\newblock Order out of chaos: Proving linearizability using local views.
\newblock {\em CoRR}, abs/1805.03992, 2018.
\newblock URL: \url{http://arxiv.org/abs/1805.03992}, \href
  {http://arxiv.org/abs/1805.03992} {\path{arXiv:1805.03992}}.

\bibitem{fraser-phd}
Keir Fraser.
\newblock {\em {Practical lock-freedom}}.
\newblock PhD thesis, {University of Cambridge, Computer Laboratory},
  University of Cambridge, Computer Laboratory, February 2004.

\bibitem{DBLP:conf/concur/2004}
Philippa Gardner and Nobuko Yoshida, editors.
\newblock {\em {CONCUR} 2004 - Concurrency Theory, 15th International
  Conference, London, UK, August 31 - September 3, 2004, Proceedings}, volume
  3170 of {\em Lecture Notes in Computer Science}. Springer, 2004.
\newblock URL: \url{https://doi.org/10.1007/b100113}, \href
  {http://dx.doi.org/10.1007/b100113} {\path{doi:10.1007/b100113}}.

\bibitem{HarrisList}
Timothy~L. Harris.
\newblock {A Pragmatic Implementation of Non-blocking Linked-Lists}.
\newblock In {\em DISC}, 2001.

\bibitem{HellerHLMMS05}
Steve Heller, Maurice Herlihy, Victor Luchangco, Mark Moir, Bill Scherer, and
  Nir Shavit.
\newblock A lazy concurrent list-based set algorithm.
\newblock In {\em OPODIS}, 2005.

\bibitem{conf/concur/HenzingerSV13}
Thomas~A. Henzinger, Ali Sezgin, and Viktor Vafeiadis.
\newblock Aspect-oriented linearizability proofs.
\newblock In {\em CONCUR}, pages 242--256, 2013.

\bibitem{TOPLAS:HW90}
M.~P. Herlihy and J.~M. Wing.
\newblock Linearizability: a correctness condition for concurrent objects.
\newblock {\em ACM Transactions on Programming Languages and Systems}, 12(3),
  1990.

\bibitem{OptSkipList}
Maurice Herlihy, Yossi Lev, Victor Luchangco, and Nir Shavit.
\newblock {A Simple Optimistic Skiplist Algorithm}.
\newblock In {\em SIROCCO}, 2007.

\bibitem{TAOMPP}
Maurice Herlihy and Nir Shavit.
\newblock {\em The Art of Multiprocessor Programming}.
\newblock Morgan Kaufmann Publishers Inc., San Francisco, CA, USA, 2008.

\bibitem{Howley:2012}
Shane~V. Howley and Jeremy Jones.
\newblock {A Non-blocking Internal Binary Search Tree}.
\newblock In {\em SPAA}, 2012.

\bibitem{DBLP:conf/ifip/Jones83}
Cliff~B. Jones.
\newblock Specification and design of (parallel) programs.
\newblock In {\em {IFIP} Congress}, pages 321--332, 1983.

\bibitem{DBLP:conf/popl/JungSSSTBD15}
Ralf Jung, David Swasey, Filip Sieczkowski, Kasper Svendsen, Aaron Turon, Lars
  Birkedal, and Derek Dreyer.
\newblock Iris: Monoids and invariants as an orthogonal basis for concurrent
  reasoning.
\newblock In Sriram~K. Rajamani and David Walker, editors, {\em Proceedings of
  the 42nd Annual {ACM} {SIGPLAN-SIGACT} Symposium on Principles of Programming
  Languages, {POPL} 2015, Mumbai, India, January 15-17, 2015}, pages 637--650.
  {ACM}, 2015.
\newblock URL: \url{http://doi.acm.org/10.1145/2676726.2676980}, \href
  {http://dx.doi.org/10.1145/2676726.2676980}
  {\path{doi:10.1145/2676726.2676980}}.

\bibitem{DBLP:journals/pacmpl/KrishnaSW18}
Siddharth Krishna, Dennis~E. Shasha, and Thomas Wies.
\newblock Go with the flow: compositional abstractions for concurrent data
  structures.
\newblock {\em {PACMPL}}, 2({POPL}):37:1--37:31, 2018.
\newblock URL: \url{http://doi.acm.org/10.1145/3158125}, \href
  {http://dx.doi.org/10.1145/3158125} {\path{doi:10.1145/3158125}}.

\bibitem{DBLP:conf/wdag/Lev-AriCK15}
Kfir Lev{-}Ari, Gregory~V. Chockler, and Idit Keidar.
\newblock A constructive approach for proving data structures' linearizability.
\newblock In Yoram Moses, editor, {\em Distributed Computing - 29th
  International Symposium, {DISC} 2015, Tokyo, Japan, October 7-9, 2015,
  Proceedings}, volume 9363 of {\em Lecture Notes in Computer Science}, pages
  356--370. Springer, 2015.
\newblock URL: \url{https://doi.org/10.1007/978-3-662-48653-5_24}, \href
  {http://dx.doi.org/10.1007/978-3-662-48653-5_24}
  {\path{doi:10.1007/978-3-662-48653-5_24}}.

\bibitem{DBLP:conf/popl/Ley-WildN13}
Ruy Ley{-}Wild and Aleksandar Nanevski.
\newblock Subjective auxiliary state for coarse-grained concurrency.
\newblock In Roberto Giacobazzi and Radhia Cousot, editors, {\em The 40th
  Annual {ACM} {SIGPLAN-SIGACT} Symposium on Principles of Programming
  Languages, {POPL} '13, Rome, Italy - January 23 - 25, 2013}, pages 561--574.
  {ACM}, 2013.
\newblock URL: \url{http://doi.acm.org/10.1145/2429069.2429134}, \href
  {http://dx.doi.org/10.1145/2429069.2429134}
  {\path{doi:10.1145/2429069.2429134}}.

\bibitem{DBLP:conf/pldi/LiangF13}
Hongjin Liang and Xinyu Feng.
\newblock Modular verification of linearizability with non-fixed linearization
  points.
\newblock In {\em {ACM} {SIGPLAN} Conference on Programming Language Design and
  Implementation, {PLDI} '13, Seattle, WA, USA, June 16-19, 2013}, pages
  459--470, 2013.

\bibitem{MichaelList}
Maged~M. Michael.
\newblock {High Performance Dynamic Lock-free Hash Tables and List-based Sets}.
\newblock In {\em SPAA}, 2002.

\bibitem{Natarajan:2014}
Aravind Natarajan and Neeraj Mittal.
\newblock {Fast Concurrent Lock-free Binary Search Trees}.
\newblock In {\em PPoPP}, 2014.

\bibitem{PODC10:Hindsight}
P.~W. O'Hearn, N.~Rinetzky, M.~T. Vechev, E.~Yahav, and G.~Yorsh.
\newblock Verifying linearizability with hindsight.
\newblock In {\em 29th Annual ACM SIGACT-SIGOPS Symposium on Principles of
  Distributed Computing (PODC)}, pages 85--94, 2010.

\bibitem{DBLP:conf/concur/OHearn04}
Peter~W. O'Hearn.
\newblock Resources, concurrency and local reasoning.
\newblock In Gardner and Yoshida \cite{DBLP:conf/concur/2004}, pages 49--67.
\newblock URL: \url{https://doi.org/10.1007/978-3-540-28644-8_4}, \href
  {http://dx.doi.org/10.1007/978-3-540-28644-8_4}
  {\path{doi:10.1007/978-3-540-28644-8_4}}.

\bibitem{DBLP:journals/cacm/OwickiG76}
Susan~S. Owicki and David Gries.
\newblock Verifying properties of parallel programs: An axiomatic approach.
\newblock {\em Commun. {ACM}}, 19(5):279--285, 1976.
\newblock URL: \url{http://doi.acm.org/10.1145/360051.360224}, \href
  {http://dx.doi.org/10.1145/360051.360224} {\path{doi:10.1145/360051.360224}}.

\bibitem{DBLP:conf/popl/ParkinsonBO07}
Matthew~J. Parkinson, Richard Bornat, and Peter~W. O'Hearn.
\newblock Modular verification of a non-blocking stack.
\newblock In Martin Hofmann and Matthias Felleisen, editors, {\em Proceedings
  of the 34th {ACM} {SIGPLAN-SIGACT} Symposium on Principles of Programming
  Languages, {POPL} 2007, Nice, France, January 17-19, 2007}, pages 297--302.
  {ACM}, 2007.
\newblock URL: \url{http://doi.acm.org/10.1145/1190216.1190261}, \href
  {http://dx.doi.org/10.1145/1190216.1190261}
  {\path{doi:10.1145/1190216.1190261}}.

\bibitem{intlf}
Arunmoezhi Ramachandran and Neeraj Mittal.
\newblock {A Fast Lock-Free Internal Binary Search Tree}.
\newblock In {\em ICDCN}, 2015.

\bibitem{DBLP:conf/esop/SergeyNB15}
Ilya Sergey, Aleksandar Nanevski, and Anindya Banerjee.
\newblock Specifying and verifying concurrent algorithms with histories and
  subjectivity.
\newblock In Jan Vitek, editor, {\em Programming Languages and Systems - 24th
  European Symposium on Programming, {ESOP} 2015, Held as Part of the European
  Joint Conferences on Theory and Practice of Software, {ETAPS} 2015, London,
  UK, April 11-18, 2015. Proceedings}, volume 9032 of {\em Lecture Notes in
  Computer Science}, pages 333--358. Springer, 2015.
\newblock URL: \url{https://doi.org/10.1007/978-3-662-46669-8_14}, \href
  {http://dx.doi.org/10.1007/978-3-662-46669-8_14}
  {\path{doi:10.1007/978-3-662-46669-8_14}}.

\bibitem{DBLP:journals/tods/ShashaG88}
Dennis~E. Shasha and Nathan Goodman.
\newblock Concurrent search structure algorithms.
\newblock {\em {ACM} Trans. Database Syst.}, 13(1):53--90, 1988.
\newblock URL: \url{http://doi.acm.org/10.1145/42201.42204}, \href
  {http://dx.doi.org/10.1145/42201.42204} {\path{doi:10.1145/42201.42204}}.

\bibitem{absoluteness}
Joseph~R Shoenfield.
\newblock The problem of predicativity.
\newblock In {\em Mathematical Logic In The 20th Century}, pages 427--434.
  World Scientific, 2003.

\bibitem{Triplett:2011:RSC}
Josh Triplett, Paul~E. McKenney, and Jonathan Walpole.
\newblock {Resizable, Scalable, Concurrent Hash Tables via Relativistic
  Programming}.
\newblock In {\em USENIX ATC}, 2011.

\bibitem{DBLP:conf/icfp/TuronDB13}
Aaron Turon, Derek Dreyer, and Lars Birkedal.
\newblock Unifying refinement and hoare-style reasoning in a logic for
  higher-order concurrency.
\newblock In Greg Morrisett and Tarmo Uustalu, editors, {\em {ACM} {SIGPLAN}
  International Conference on Functional Programming, ICFP'13, Boston, MA,
  {USA} - September 25 - 27, 2013}, pages 377--390. {ACM}, 2013.
\newblock URL: \url{http://doi.acm.org/10.1145/2500365.2500600}, \href
  {http://dx.doi.org/10.1145/2500365.2500600}
  {\path{doi:10.1145/2500365.2500600}}.

\bibitem{phd/Vafeiadis08}
V.~Vafeiadis.
\newblock {\em Modular fine-grained concurrency verification}.
\newblock PhD thesis, University of Cambridge, 2008.

\bibitem{PPOPP:VafeiadisHHS06}
V.~Vafeiadis, M.~Herlihy, T.~Hoare, and M.~Shapiro.
\newblock Proving correctness of highly-concurrent linearisable objects.
\newblock In {\em PPoPP}, 2006.

\bibitem{conf/cav/Vafeiadis10}
Viktor Vafeiadis.
\newblock Automatically proving linearizability.
\newblock In {\em CAV '10}, volume 6174 of {\em LNCS}, pages 450--464.

\bibitem{conf/vmcai/Vafeiadis09}
Viktor Vafeiadis.
\newblock Shape-value abstraction for verifying linearizability.
\newblock In {\em VMCAI '09: Proc. 10th Intl. Conf. on Verification, Model
  Checking, and Abstract Interpretation}, volume 5403 of {\em LNCS}, pages
  335--348. Springer, 2009.

\bibitem{conf/ppopp/VafeiadisHHS06}
Viktor Vafeiadis, Maurice Herlihy, Tony Hoare, and Marc Shapiro.
\newblock Proving correctness of highly-concurrent linearisable objects.
\newblock In {\em PPOPP '06}, pages 129--136. ACM.

\bibitem{LazyListSafety}
Viktor Vafeiadis, Maurice Herlihy, Tony Hoare, and Marc Shapiro.
\newblock A safety proof of a lazy concurrent list-based set implementation.
\newblock Technical Report UCAM-CL-TR-659, University of Cambridge, Computer
  Laboratory, 2006.

\bibitem{DBLP:conf/cav/ZhuPJ15}
He~Zhu, Gustavo Petri, and Suresh Jagannathan.
\newblock Poling: {SMT} aided linearizability proofs.
\newblock In {\em Computer Aided Verification - 27th International Conference,
  {CAV} 2015, San Francisco, CA, USA, July 18-24, 2015, Proceedings, Part
  {II}}, pages 3--19, 2015.

\end{thebibliography}
 
\iflong

%
%
 
\clearpage

\appendix

\section{Proof Method}\label{Se:AbstractFrameworkLong}



In this section, we describe our main result, which allows to lift a property $\Pred$ that holds on the local view of a thread accumulated by performing a sequence of reads $\seqr$ to a property that held on the global, concrete, state \emph{at some point} during the execution of $\seqr$.

\paragraph{Programming model.}
A \emph{global state} (state) is a mapping between \emph{memory locations} (locations)  and \emph{values}.
A value is either a natural number, a location, or   $\NULL$.
Without loss of generality, we assume threads share access to immutable global variables and to a mutable \emph{global heap}.
Thus, memory locations are used to stores the values of fields of objects.
A \emph{concurrent  execution} (execution) $\pi$ is a sequence of states produced by an interleaving of atomic actions issued by threads.
A pair of states $(H,H')$ is a \emph{transition} in $\pi$ if $\pi=\cdots H H' \cdots$\,.
Without loss of generality, we assume that  each transition results from either a \emph{read} or \emph{write} operation.
(We treat synchronization actions, e.g., \emph{lock} and \emph{unlock}, as writes.)
A \emph{read} $r$ consists of a value $v$ and a location $\readr{r}$ with the meaning that $r$ reads $v$ from $\readr{r}$.
Similarly, a write $w$ consists of a value $v$ and a location $\modw{w}$ with the meaning that $w$ sets $\modw{w}$ to $v$.
We denote by $w(H)$ the state resulting from the execution of $w$ on state $H$.

\subsection{Between the Local View and the Global State}
\ignore{
In this section, we describe our main result, which allows to lift a property $\Pred$ that holds on the local view of a thread accumulated by performing a sequence of reads $\seqr$ to a property that held on the global, concrete, state \emph{at some point} during the execution of $\seqr$.
We start with a few definitions that will let us formalize this claim as well as the conditions it requires.
}

\ignore{
\paragraph{Writes.}
The conditions under which our main claim holds refer to sequences of \emph{writes} performed on the global state by other threads that run concurrently with the thread executing the reading sequence $\seqr$.
A \emph{state} is a mapping between locations and values, values being numeric or locations.
A \emph{write} $w = (\ell, v, \Prew{})$ consists of a value $v$, a location $\ell$ and a precondition $\Prew{}$, with the meaning that the value $v$ is assigned to location $\ell$. The write is considered \emph{valid} on a state $H$ if the precondition holds in that state: $\Prew{}(H)$ is true. (In our setting, preconditions are given by the assertions in the code that precede a statement performing a write.) We write $\modw{w}$ to denote the location $\ell$ that $w$ writes to, $\Prew{w}$ to denote its precondition, and $w(H)$ to denote the state resulting from the execution of $w$ on state $H$.
}


For the rest of this section, we fix a sequence $\seqr= r_1,\ldots,r_d$ of reads performed by a thread, and the execution $\pi$ from which it is taken.
We denote the (global) state when the reading sequence $\seqr$ started its execution by $\hinit$.
We denote the sequence of writes performed concurrently with $\seqr$ in $\pi$ by $\seqw = w_1,\ldots,w_n$.
%
The sequence $\seqw$ produces intermediate \emph{global states} after the execution of each write.
We denote the global state after execution of $w_1 \ldots w_i$ by $\opprefix{\hconc}{i}$, i.e., $\opprefix{\hconc}{i} = w_1 \ldots w_i(\hinit)$.


\para{Local views} 
The sequence of reads $\seqr$ induces a state $\hlocal$, which directly corresponds to the values $\seqr$ observes in memory.
This is constructed by assigning to location $x$ the value read in the last read in $\seqr$ if $x$ is read at all, i.e.,
When $\seqr$ starts its local view $\hlocali{0}$ is empty, and, assuming its $i$th read is $(\loc,v)$, the
produced local view is $\hlocali{i}=\hlocali{i-1}[\loc \mapsto v]$.
We refer to $\hlocal=\hlocali{d}$ as the \emph{local view} produced by $\seqr$ (\emph{local view} for short).
We emphasize that while technically $\hlocal$ is a state, it is \emph{not} necessarily an actual intermediate global state, and may have never existed in memory during the execution. 


%

\yotam{Write somewhere (in intro?) that we explain the local view by a fabricated heap, if you can prove something for fabricated heaps that you know this for the local view in the ``sequential reasoning''}

\para{Fabricated state}
In order to relate between the local view and the global state, we consider a \emph{fabricated state} $\hobserved$. The fabricated state is a state such that (i)~$\hlocal \subseteq \hobserved$, and (ii)~$\hobserved$ is constructed by a \emph{subsequence} $\seqwf = w_{i_1},\ldots,w_{i_k}$ of
$\seqw$, i.e., $\hobserved = w_{i_1} \ldots w_{i_k}(\hinit)$.
Again,  while the sequence of writes $\seqw = w_1,\ldots,w_n$ comes from an execution, the subsequence $w_{i_1},\ldots,w_{i_k}$ may not be produced by any execution. and thus $\hobserved$ may not occur in any execution of the algorithm. 
As before, we denote the intermediate fabricated state constructed after performing the first $j$ writes in $\seqwf$ 
by $\opprefix{\hobserved}{j}$, i.e., $\opprefix{\hobserved}{j} = w_{i_1} \ldots w_{i_j}(\hinit)$.
The notion of a fabricated state is in line with the intuition that the reading sequence has been affected by some of the concurrent writes but also missed some. The specific construction of the fabricated state we use in our framework is provided in \Cref{subsec:fabricated}.



Our proof approach (i) establishes a connection, that we call \emph{simulation}, between the effect of the writes 
on the fabricated state to their effect in the concrete state, and (ii) uses a property called \emph{upward absoluteness} to relate the local view to the fabricated state (based on the property that $\hlocal \subseteq \hobserved$).

In this way, the fabricated state allows us to consider the memory ($\hlocal$) that $\seqr$ observes as if $\seqr$ were operating sequentially on a state ($\hobserved$) that is closely related to the global memory state ($\hconc$).

\subsubsection{From the Fabricated State to the Global State Using Simulation}
The fabricated state will be constructed such that it has the following connection to the concrete state:
if a write turns $\Pred$ to true in an intermediate fabricated state $\opprefix{\hobserved}{j} = w_{i_1} \ldots w_{i_j}(\hinit)$ obtained after executing $w_{i_1},\ldots,w_{i_j}$, then it also turns $\Pred$ to true in the intermediate global state $\opprefix{\hconc}{i_j} = w_{1} \ldots w_{i_j}(\hinit)$ obtained after executing the prefix  $w_{1},\ldots,w_{i_j}$ of the full sequence of writes. This is formalized in the following definition. We then show that if 
this connection is established, then $\Pred$ being true in the fabricated state transfers to it being true in \emph{some} intermediate point in the \emph{global} state.
%

\begin{definition}[Simulation]
\label{def:uniform-change-abstract}
The subsequence of writes $w_{i_1} \ldots w_{i_k}$ \emph{simulates} the sequence $w_1 \ldots w_n$ w.r.t.\ $\Pred$ if for every $1 \leq j \leq k$,
if $\neg\Pred(\opprefix{\hobserved}{j-1})$ but $\Pred(w_{i_j}(\opprefix{\hobserved}{j-1}))$, then
$\neg\Pred(\opprefix{\hconc}{i_{j}-1}) \implies \Pred(w_{i_j}(\opprefix{\hconc}{i_{j}-1}))$.

We say that $\hobserved = w_{i_1} \ldots w_{i_k}(\hinit)$ \emph{simulates} $\hconc = w_1 \ldots w_n(\hinit)$ w.r.t.\ $\Pred$ if $w_{i_1} \ldots w_{i_k}$ simulates $w_1 \ldots w_n$ w.r.t.\ $\Pred$.
\todobom{A drawing would be useful here}
\todobom{$\Pred$-simulates instead of simulates w.r.t. $\Pred$?}
\end{definition}

\begin{lemma} \label{lem:simulation}
If $\hobserved$ simulates $\hconc$ w.r.t.\ $\Pred$ and $\Pred(\hobserved)$ holds, then there exists some $0 \leq i \leq n$ s.t.\ $\Pred(\opprefix{\hconc}{i})$.
\end{lemma}
\iflong
\begin{proof}
	If $\Pred(\hinit)$, then $i = 0$ establishes the claim.
	Otherwise, let $w_{i_j}$ be the first write to make $\Pred$ true in $\hobserved$, namely: let $1 \leq j \leq k$ be the minimal $j$ such that $\neg\Pred(\opprefix{\hobserved}{j-1}$ but $\Pred(\opprefix{\hobserved}{j})$.
	If $\Pred(\opprefix{\hconc}{i_j-1})$, take $i = i_j - 1$. Otherwise, $\neg\Pred(\opprefix{\hconc}{i_j-1})$.
	So we have: $\neg\Pred(\opprefix{\hobserved}{j-1})$ but $\Pred(w_{i_j}(\opprefix{\hobserved}{j-1}))$ and $\neg\Pred(\opprefix{\hconc}{i_j - 1})$.
	From the premise that $\hobserved$ simulates $\hconc$ it follows that $\Pred(w_{i_j}(\opprefix{\hconc}{i_j - 1}))$, and by taking $i = i_j$ the claim follows.
\end{proof}
\fi

\begin{remark}
[Information Between Predicates via Hindsight]
In certain cases, it is useful to consider more than one predicate on the fabricated state.
For example, when $y$ is inserted as a child of $x$, the reachability of $y$ is determined by the reachability of $x$ before the write.
To this end, we prove the translation from the fabricated state to the global state for all predicates of interest simultaneously.
We use this approach in~\Cref{subsec:search-paths}.
\end{remark}

\subsubsection{From the Local View to the Fabricated State Using Upward-Absoluteness}

Simulation lets us relate the fabricated state to the global state. Next, we relate the value of $\Pred$ on the local state to its value on the fabricated state.
Recall that the fabricated state does not correspond exactly to the local view. Instead, $\hlocal \subseteq \hobserved$.
This gap is bridged when $\Pred$ is \emph{upward-absolute}~\cite{absoluteness}:
\begin{definition}[Upward-Absoluteness]
\label{def:upward-absolutenss}
A predicate $\Pred$ is \emph{upward-absolute} if for every pair of states $H,H'$ such that $H \subseteq H'$, $\Pred(H) \implies \Pred(H')$.
\end{definition}
Upward absoluteness is in line with the fact that an operation rarely observes all the contents of all memory locations.
The procedure may have read only a partial view of memory. However, the decision based on $\Pred(\hlocal)$ must hold regardless of unobserved locations.

For example, $\QReachXK{x}{k}$, for any location $x$ and key $k$,
is upward-absolute, because if a memory state contains a path then so does every extension of this state.

As we take $\hobserved$ such that $\hlocal \subseteq \hobserved$, we deduce from upward-absoluteness that if $\Pred(\hlocal)$ then $\Pred(\hobserved)$ holds, and, if $\hobserved$ simulates $\hconc$, then $\Pred(\opprefix{\hconc}{i})$ holds for some $i$. This is summarized in the following theorem.

\begin{theorem}
If $\Pred$ is upward-absolute, $\Pred(\hlocal)$ holds, $\hlocal \subseteq \hobserved$, and $\hobserved$ simulates $\hconc$ w.r.t. $\Pred$, then there exists some $0 \leq i \leq n$ s.t.\ $\Pred(\opprefix{\hconc}{i})$.
\end{theorem}

\subsection{The Fabricated State} \label{subsec:fabricated}

In this section, we define the fabricated state $\hobserved$, obtained by a subsequence of writes $w_{i_1},\ldots,w_{i_k}$. Our goal is to ensure that $\hlocal \subseteq \hobserved$ and
to relate $\hobserved$ to $\hconc$ by simulation (\Cref{def:uniform-change-abstract}).

In order to ensure that $\hlocal \subseteq \hobserved$, it suffices to include in $w_{i_1},\ldots,w_{i_k}$ all the writes that affected the reads that constructed the local view. However, ensuring simulation requires a more involved construction.
%
%
\Cref{def:uniform-change-abstract} requires writes on the fabricated state to have a similar effect as in the global state. To achieve this, we choose the fabricated state to be as similar as possible to the global state from the perspective of the write and the predicate it might affect.

To this end, we develop a consistency condition of \emph{forward-agreement}, based on a
\emph{partial order} on memory. The idea is that for each write $w_{i_j}$ in $w_{i_1},\ldots,w_{i_k}$, the write sees the same picture of the memory that comes forward in the data structure both in the fabricated state $\opprefix{\hobserved}{j-1}$ and in the corresponding global state $\opprefix{\hconc}{i_j-1}$. We show how to construct the fabricated state so that this property holds.
While this property alone does not suffice to imply simulation, in \Cref{subsec:search-paths} we show how together with an additional property (\Cref{def:preservation}), forward agreement can be used to prove the simulation property w.r.t. predicates that track reachability along search paths.


\para{Order on memory} Every state of the data structure induces a certain order on how operations read the different memory locations.
To capture this, the user provides a mapping from a state $H$ to a \emph{partial order} that $H$ induces on memory locations, denoted $\leqloc_H$.
%
%

\begin{example}\label{Ex:OrderSearchPaths}  
In the running example of \Cref{Fi:Running},
the order $\leqloc_H$ on memory locations , i.e., fields of objects, is defined by following pointers from parent to children, i.e., all the fields of $x.\textit{left}$ and $x.\textit{right}$ are ordered after the fields of $x$, and the fields of an object are ordered by $x.key < \QXD{x} < \{x.\textit{left}, x.\textit{right}\}$. \yf{a drawing would be nice}
\end{example}

\ignore{
As $\leqloc_H$ depends on $H$,
it changes with time (as the state changes).
In the construction of the fabricated state, we consider the \emph{accumulated order}, defined as follows:
}

We note that as $\leqloc_H$ depends on $H$,
it changes with time (as the state changes).
We make the following requirements on $\leqloc_H$, which the user also needs to establish (for \Cref{lem:forward-agreement,lem:inclusion} below).

\para{Locality of the order} In order to ensure that the fabricated state is forward agreeing with the global state (\Cref{lem:forward-agreement}), we require that
the order $\leqloc_{H}$ is determined \emph{locally} in the sense that if $\ell_2$ is an immediate successor of $\ell_1$ in $\leqloc_H$, then for every $H'$ such that $H'(\ell_1) = H(\ell_1)$ it holds that $\ell_1 \leqloc_{H'} \ell_2$. Note that the value in the target location does not affect the inclusion in the order; as an illustration, if $\leqloc_H$ is defined based on pointers, and $\ell_1$ is a pointer to $\ell_2$, making $\ell_2$ an immediate successor of $\ell_1$ in $\leqloc_H$, then this definition depends on the value in $\ell_1$ but not on the value in $\ell_2$.

\para{Read in order} To enforce the fact that the order captures the order in which operations read the different memory locations, we require reads to the local view to respect the order. Formally, consider a read $r$ in the sequence $\seqr$ reading the  location $\ell$ from the global state $\opprefix{\hconc}{m} = w_1 \ldots w_m(\hinit)$, and let $\readset{}$ be the set of locations read by earlier reads in $\seqr$. We require $\readset{} \leqacculoc_{w_1 \ldots w_m(\hinit)} \ell$.
The read-in-order property holds if reads always read a location further in the order in the current global state. Namely, if $\ell'$ is the last location read into the local view, the next location $\ell$ read is such that $\ell' \leqloc_{\opprefix{\hconc}{m}} \ell$.

Note that the read must respect the order of the global state but is oblivious to most of the global state; this is in line with the local nature of the order.

\begin{example}
The fact that the read-in-order property holds for all the methods in~\Cref{Fi:Running} follows from a very simple syntactic analysis, e.g., in the case of \code{locate(k)}, children are always read after their parents and the field $\textit{key}$ is always accessed before $\textit{left}$ or $\textit{right}$.
\end{example}

\para{Acyclicity of the accumulated order} We define the \emph{accumulated order} w.r.t.\ a sequence of writes $\hat{w}_1,\ldots,\hat{w}_m$, denoted $\leqacculoc_{\hat{w}_1 \ldots \hat{w}_m(\hinit)}$, as the transitive closure of
$\bigcup\limits_{0 \leq s \leq m} \leqloc_{\hat{w}_1 \ldots \hat{w}_s(\hinit)}$.
We require that the accumulated order $\leqacculoc_{w_1 \ldots w_n(\hinit)}$ is a partial order.

\begin{example}
In the running example of \Cref{Fi:Running},
the accumulated order is constructed by collecting all pointer links created during an execution.
As explained in \Cref{example:acyclicity-cf} this relation is acyclic, and hence remains a partial order.
\end{example}

\para{Construction of the fabricated state}
We exploit the order to construct $\hobserved$ so that it contains $\hlocal$, but also satisfies the property of forward-agreement with the global state.
Formally, we construct $\hobserved$ in the following way: consider the reads that form $\hlocal$. Each read $r$ in $\seqr$ is of some location $x_r$ in an intermediate memory state $\opprefix{\hconc}{m} = w_1 \ldots w_m(\hinit)$. We take the writes that occurred \emph{backwards} in time and modify locations \emph{forward} in the accumulated order:
$\mathop{\textit{precede-forward}}(r) = \{w_j \mid j \leq m \land x_r \leqacculoc_{w_1 \ldots w_m(\hinit)} \modw{w_j}\}$.
The subsequence $\seqwf = w_{i_1},\ldots,w_{i_k}$ is taken to be the union of $\mathop{\textit{precede-forward}}(r)$ for all $r$'s in $\seqr$.

\ignore{
Roughly speaking, taking all the forward writes ensures that $\hobserved$ is forward-agreeing with $\hconc$ (under an assumption of \emph{locality} of the order).
\yf{rephrase following}
Furthermore, under an assumption that reads respect the order and that the accumulated order is still a partial order,
we have that $\hlocal \subseteq \hobserved$. This is because in each read step, the writes that modify locations forward in the order compared to the current read are consistent with the view of the thread to this point.
We formalize these claims next.
}


\subsubsection{Forward-Agreement as a Step Toward Simulation}\label{subsec:forward-agreement}
Next we formalize \emph{forward-agreement} between the fabricated state and the global state. Forward-agreement requires that when a write in the subsequence $w_{i_1},\ldots,w_{i_k}$ is performed on the intermediate fabricated state, the locations that come \emph{forward} in the order --- in the order induces by the before or after the writ --- have \emph{exactly} the same values as when the write is performed in global memory.

\sharon{Perhaps better to phrase this using the accumulated order}
\begin{definition}[Forward-Agreement]
\label{def:forward-agreement}
	$\hobserved = w_{i_1} \ldots w_{i_k}(\hinit)$ is \emph{forward-agreeing} with $\hconc = w_1 \ldots w_n(\hinit)$ if for every $0 \leq j < k$,
	\begin{equation}
		\forall \ell. \ 	\left( \modw{w_{i_j}} \leqloc_{\opprefix{\hobserved}{j-1}} \ell \right)
					  \lor
					  		\left( \modw{w_{i_j}} \leqloc_{\opprefix{\hobserved}{j}} \ell \right)
						\implies
							\opprefix{\hobserved}{j-1}(\ell) = \opprefix{\hconc}{i_j-1}(\ell).
	\end{equation}
\end{definition}


	\begin{lemma} \label{lem:forward-agreement}
		If the order $\leqloc_H$ is determined locally, then $\hobserved$ is forward-agreeing with $\hconc$.
	\end{lemma}

\iflong
For the proof of \Cref{lem:forward-agreement}, we first show that, under the locality assumption, the accumulated order induced by a subsequence of writes is contained in the accumulated order of the entire sequence (this is important for relating the order in the fabricated and global state).
We note that the accumulated order is trivially monotonic w.r.t.\ additional writes: $\leqacculoc_{w_1 \ldots w_i(\hinit)} \ \subseteq \ \leqacculoc_{w_1 \ldots w_{i+1}(\hinit)}$, but when considering a subsequence, the intermediate states do not coincide.

\begin{lemma}
	If the order $\leqloc_H$ is determined locally, then for every sequence of writes $w_1,\ldots,w_n$ and every subsequence $w_{i_1},\ldots,w_{i_k}$ operating on a state $\hinit$,
	\begin{equation*}
		\leqacculoc_{w_{i_1} \ldots w_{i_k}(\hinit)}  \ \subseteq \ \leqacculoc_{w_1 \ldots w_n(\hinit)}
	\end{equation*}
\end{lemma}
\begin{proof}
Assume $\ell_1 \leqacculoc_{w_{i_1} \ldots w_{i_n}(\hinit)} \ell_2$, and let $x_1,\ldots,x_m$ be a sequence of locations such that $x_1 = \ell_1$, $x_m = \ell_2$ and for every $0 \leq i < m$, $x_i \leqloc_{w_{i_1} \ldots w_{i_s}(\hinit)} x_{i+1}$ for some $0 \leq s \leq k$. It suffices to show that $\leqloc_{w_{i_1} \ldots w_{i_s}(\hinit)}  \ \subseteq \ \leqacculoc_{w_1 \ldots w_n(\hinit)}$, because this implies that $x_1,\ldots,x_m$ is an appropriate sequence for establishing that $\ell_1 \leqacculoc_{w_1 \ldots w_n(\hinit)} \ell_2$.

Assume therefore that $\ell_1 \leqloc_{w_{i_1} \ldots w_{i_k}(\hinit)} \ell_2$, and prove that $\ell_1 \leqacculoc_{w_1 \ldots w_n(\hinit)} \ell_2$.
Since the order is discrete, there is a sequence $x_1,\ldots,x_m$ of locations such that $x_1 = \ell_1$, $x_m = \ell_2$ and $x_{i+1}$ is the immediate successor of $x_i$ in $\leqloc_{w_{i_1} \ldots w_{i_k}(\hinit)}$ for all $0 \leq i < m$. It suffices to show that $x_i \leqacculoc_{w_1 \ldots w_n(\hinit)} x_{i+1}$ for every $0 \leq i < m$.

Fix some $i$. Since $w_{i_1},\ldots,w_{i_k}$ is a subsequence of $w_1,\ldots,w_n$, there is some intermediate state $w_1 \ldots w_s(\hinit)$ that agrees on the value in the \emph{single location} $x_i$, meaning $(w_{i_1} \ldots w_{i_k}(\hinit))(x_i) = (w_1 \ldots w_s(\hinit))(x_i)$.
By the locality assumption, $x_i \leqloc_{w_1 \ldots w_s(\hinit)} x_{i+1}$. This implies that $x_i \leqacculoc_{w_1 \ldots w_n(\hinit)} x_{i+1}$.
The claim follows.
\end{proof}

	\begin{proof}[Proof of \Cref{lem:forward-agreement}]
		Let $w_{i_j}$ be a write in the subsequence. Let $\ell$ be a location s.t.\ $\modw{w_{i_j}} \leqloc_{\opprefix{\hobserved}{j-1}} \ell$ or $\modw{w_{i_j}} \leqloc_{\opprefix{\hobserved}{j}} \ell$. We need to show that $\opprefix{\hobserved}{j-1}(\ell) = \opprefix{\hconc}{i_j - 1}(\ell)$.

		Let $w_s$ the last write (in the global memory) to modify $\ell$ before $w_{i_j}$: take $s$ to be maximal index such that $s < i_j$ and $\modw{w_s} = \ell$.
		If there is no such $s$, $\opprefix{\hconc}{i_j - 1}(\ell) = \opprefix{\hobserved}{j-1}(\ell) = \hinit(\ell)$, since $\ell$ is not modified by any preceding write.
		Otherwise, it suffices to show that $w_s$ is also included in the subsequence, since $w_s$ is performed on both $\hobserved$ and $\hconc$ and there are no writes after $w_s$ and before $w_{i_j}$ to modify $\ell$ in the sequence of writes and therefore also in the subsequence.

		By the construction of the subsequence, there is a read $r$ of location $x_r$ from the global state $\opprefix{\hconc}{m}$ such that $i_j \leq m$ and $x_r \leqacculoc_{\opprefix{\hconc}{m}} \modw{w_{i_j}}$.
		We have that $s < m$ and $x_r \leqacculoc_{\opprefix{\hconc}{m}} \modw{w_s}$ as $\modw{w_s} = \ell$ and $\modw{w_{i_j}} \leqacculoc_{\opprefix{\hconc}{m}} \ell$ because $\modw{w_{i_j}} \leqacculoc_{\opprefix{\hconc}{i_j}} \ell$ (the accumulated order satisfies $\leqacculoc_{\opprefix{\hconc}{i_j}} \subseteq \leqacculoc_{\opprefix{\hconc}{m}}$ as $i_j \leq m$).
		Therefore, the construction includes $w_s$ in the subsequence as well.
	\end{proof}
\fi

In the next section we use forward agreement, together with an additional property, to show simulation for predicates defined by reachability along search paths.

\subsubsection{Inclusion of the Local View in the Fabricated State}

Next, we turn to showing that $\hlocal \subseteq \hobserved$.
For this purpose, we require that the definition of the order correctly captures the order of manipulation of the data structure.
Namely, we require that the reads are performed in the order dictated by $\leqloc_H$, and that writes preserve the order in the sense that they do not introduce cycles in the accumulated order.

	\begin{lemma} \label{lem:inclusion}
		If (1) 
the accumulated order $\leqacculoc_{\bar{w}(\hinit)}$ is a partial order, and (2) every sequence of reads satisfies the read-in-order property, then $\hlocal \subseteq \hobserved$.
	\end{lemma}
\iflong
	\begin{proof}
		Let $x$ be a location in $\hlocal$, and let $r$ be the last read of $x$. Assume that $r$ reads from the global state $\opprefix{\hconc}{m}$. Let $w_s$ be the write $r$ \emph{reads-from}, namely, $w_s$ where $s$ is the maximal index such that $s \leq m$ and $\modw{w_s} = x_r$. By the construction, $w_s$ is included in the subsequence, so let $d$ be such that $i_d = s$.
		From this, $\opprefix{\hobserved}{d}(x) = \hlocal(x)$. It remains to show that later writes in the subsequence do not modify this location, namely $\modw{w_{i_j}} \neq \ell$ for all $d < j \leq k$.

		Assume $\modw{w_{i_j}} = x$ for some $j$. From the construction, $w_{i_j}$ is included in the subsequence due to some read $r'$ of location $x_{r'}$ from $\opprefix{\hconc}{m'}$ such that $i_j \leq m'$ and $x_{r'} \leqacculoc_{\opprefix{\hconc}{m'}} x$.
		\begin{itemize}
			\item If $r' \leq r$, $m' \leq m$ and thus $i_j \leq m$. Since $i_d = s$ was the maximal index so that $\modw{x_s} = x$ and $\modw{w_s} = x$, $i_j \leq i_d$.
			\item If $r' > r$, because reads respect the order, $x \leqacculoc_{\opprefix{\hconc}{m'}} x_{r'}$. Since the accumulated order is anti-symmetric, it follows that $x = x_{r'}$. But this is a contradiction to the fact that $r$ is the \emph{last} to read $x_r$.
		\end{itemize}
		The claim follows.
	\end{proof}
\fi

\subsection{Simulation w.r.t. Reachability by Search Paths} \label{subsec:search-paths}
In this section we show how to prove the simulation property for predicates defined by \emph{search paths}.
We define a \emph{preservation} property that complements forward-agreement to obtain a general proof of simulation for such predicates, showing that reachability by search paths transfers from the fabricated state to the global state.
We then extend this result for predicates defined by reachability with checking another field, which is required for some operations.

\para{Search paths}
Intuitively, a $k$-search path in state $H$ is a sequence of locations following $\leqloc_H$ that is traversed when searching for $k$.
	Formally, for a parameter $k$, a $k$-search path in a state $H$ is a sequence of locations $\ell_1,\ldots,\ell_m$, with the following requirements:
	\begin{itemize}
\item If $\ell_1,\ldots,\ell_m$ is a $k$-search path in $H$ then $\ell_i \leqloc_{H} \ell_{i+1}$ and $\ell_i \neq \ell_{i+1}$ for every $1 \leq i < m$.
		\item If $\ell_1,\ldots,\ell_m$ is a $k$-search path in $H$ and $H'$ satisfies $H'(\ell_i) = H(\ell_i)$ for all $1 \leq i < m$, then $\ell_1,\ldots,\ell_m$ is a $k$-search path in $H'$ as well, i.e., the search path depends only on the values in the locations in the sequence but the last.
		\item If $\ell_1,\ldots,\ell_m$ and $\ell_m,\ldots,\ell_{m+r}$ are both $k$-search paths in $H$, then so is $\ell_1,\ldots,\ell_m,\ldots,\ell_{m+r}$.
		\item If $\ell_1,\ldots,\ell_m$ is a $k$-search path in $H$ then so is $\ell_i,\ldots,\ell_j$ for every $1 \leq i \leq j \leq m$.
	\end{itemize}

We say that $\ell_1 \searchpath{k} \ell_2$ holds in $H$ if there exists a $k$-search path in state $H$ that starts in $\ell_1$ and ends in $\ell_2$.
\begin{example}
	$k$-search paths in the tree of \Cref{Fi:Running} are consists of sequences $\langle x.key, x.\ileft, y.key \rangle$ where $y.key$ is the address pointed to by $x.\ileft$ (meaning, the location that is the value stored in $x.\ileft$) and $x.key < k$, or
	$\langle x.key, x.\iright, y.key \rangle$ where $y.key$ is the address pointed to by $x.\iright$ and $x.key > k$. 
	This definition of $k$-search paths reproduces the definition of reachability along search paths from \Cref{Se:ProofLin}.
\end{example}

\subsubsection{Simulation From Order and Preservation}

Assuming a specific location serving as the entry point to the data structure, $\rootobj$, the predicate of reachability by a $k$-search path is
$\PredSearch{k}{x} = \rootobj \searchpath{k} x$.


\begin{definition}[Preservation]
\label{def:preservation}
We say that $\seqw$ \emph{ensures preservation of reachability by search paths} if for every $1 \leq m \leq n$, 
if for some $0 \leq i < m$, $\opprefix{\hconc}{i} \models \PredSearch{k}{\modw{w_m}}$  then $\opprefix{\hconc}{m-1} \models \PredSearch{k}{\modw{w_m}}$.
\end{definition}

We now prove the simulation w.r.t. $\PredSearch{k}{x}$ between the fabricated and global state.
We analyze \emph{all} the reachability predicates for every parameter $k$ of interest together: $\PredSearch{k}{x} \in \predset$ for every $k$ and location $x$.
\begin{lemma}
If 
$\seqw$ ensures preservation of reachability by search paths, 
then the fabricated state $\hobserved$ constructed in \Cref{subsec:fabricated} simulates $\hconc$ w.r.t.
the predicate $\PredSearch{k}{x}$ 
for every $k$ and location $x$.
\end{lemma}
We prove the simulation property by induction on the index $j$ of the  write that changes the value of $\PredSearch{k}{x}$ on the fabricated state $\hobserved$.
%
%
The proof idea is as follows.
Given a write $w$ 
that creates a $k$-search path to $x$ in the fabricated state, we construct such a path in the corresponding global state. The idea is to consider the path that $w$ creates in the intermediate fabricated state after $w$, divide it to two parts: the prefix until $\modw{w}$, and the rest of the path.
Relying on forward-agreement, the part of the path from $\modw{w}$ to $x$ is exactly the same in the corresponding global state, so we need only prove that there is also an appropriate prefix.
But necessarily there has been a $k$-search path to $\modw{w}$ in the fabricated state \emph{before} $w$, so by the induction hypothesis (hindsight) applied on $\modw{w}$, exploiting the fact that simulation up to $j-1$ implies that $\PredSearch{k}{x}(\opprefix{\hobserved}{j-1}) \implies \exists 0 \leq i \leq i_{j-1}. \ \PredSearch{k}{x}(\opprefix{\hconc}{i})$ (see \Cref{lem:simulation}),
there has been a $k$-search path to $\modw{w}$ in some intermediate global state that occurred earlier than the time of $w$.
Since $w$ writes to $\modw{w}$, the preservation property ensures that there is such a path to $\modw{w}$ in the global state also at the time of the write, and the claim follows.
\iflong
\begin{proof}
	The proof is by induction on $j$, showing mutually simulation for $w_{i_j}$ and that $\PredSearch{k}{x}(\opprefix{\hobserved}{j}) \implies \exists 0 \leq i \leq i_{j}. \ \PredSearch{k}{x}(\opprefix{\hconc}{i})$ (see \Cref{lem:simulation}).

	Let $\hconc$ such that $\hconc \not\models \Pred$, and $w$ a valid write on $\hconc$.
	Let $\hobserved$ be forward-agreeing with $\hconc$, such that $\hobserved \not\models \Pred$ but $w(\hobserved) \models \Pred$.
	Our goal is to prove that $\validw{w}(\hconc) \models \Pred$.

	Let $\pi$ be the search path in $w(\hobserved)$. If $\pi$ does not include $\modw{w}$, it is also a search path in $\hobserved$ (before the write), in contradiction to the premise that $\hobserved \not\models \Pred$.
	Let $\pi = \ell_1,\ldots,\ell_m$, and $\ell_p = \modw{w}$.

	Consider the prefix of $\pi$, $\ell_1,\ldots,\ell_{p-1},\ell_p$. This is a $k$-search path in $w(\hobserved)$ (because $\pi$ is, and by slicing), but the locations $\ell_1,\ldots,\ell_{p-1}$ are the same in $\hobserved$, so it is also a $k$-search path for in $\hobserved$, and so
	$\hobserved \models \rootobj \searchpath{k} \modw{w}$.
	Thus, there is some $i$ such that $\opprefix{\hconc}{i} \models \Pred$, and since $w$ is valid on $\hconc$, from preservation to $\modw{w}$ it follows that $\hconc \models \rootobj \searchpath{k} \modw{w}$. This is also true after the write: $\validw{w}(\hconc) \models \rootobj \searchpath{k} \modw{w}$ because the locations $\ell_1,\ldots,\ell_{p-1}$ are different from $\modw{w}$ and thus remain the same after the write.

	Now, $\modw{w} = \ell_p,\ldots,\ell_m$ is a $k$-search path in $w(\hobserved)$. As that $\modw{w} \leq_{w(\hobserved)} \ell_q$ for all $p \leq q$, from forward agreement we have that these locations have the same value also in $\validw{w}(\hconc)$. It follows that $\validw{w}(\hconc) \models \modw{w} \searchpath{k} x$.

	Since $\validw{w}(\hconc) \models \rootobj \searchpath{k} \modw{w}$ and $\validw{w}(\hconc) \models \modw{w} \searchpath{k} x$, we have that $\validw{w}(\hconc) \models \rootobj \searchpath{k} x$.
\end{proof}
\fi

We conclude:
\begin{theorem}
\label{thm:main-result-instance-reach}
If $\leqloc_H$ is determined locally, every sequence of reads satisfies the read-in-order property,
the accumulated order $\leqacculoc_{\bar{w}(\hinit)}$ is a partial order, and
$\seqw$ ensures preservation of reachability by search paths, 
then for every $k$ and location $x$, if $\PredSearch{k}{x}$ holds on $\hlocal$, then there exists
$0 \leq i \leq n$ s.t.\ $\PredSearch{k}{x}$ holds on $\opprefix{\hconc}{i}$.
\end{theorem}

\subsubsection{Reachability with Another Field}\label{Se:ReachWithField}
We now extend the previous result to predicates involving not only the reachability of a location $x$ by a search path but also a property of one of its fields.
Let $\PredXF(y)$ be a property of the value of field of an object $y$. Our goal is to establish properties of the form $\PReachXK{y}{k} \land \PredXF(y)$. As we show here, extending our results from plain reachability to reachability with an additional field is straightforward.

Such predicates are useful for the correctness of ${\tt contains}$, which locates an element and checks without locks whether it is logically deleted, e.g., $\PredXF(y) = \QXD{y}$ or $\PredXF(y) = \QXD[\neg]{y}$.
As another example, consider the predicate $\QReachXK{x}{k} \land \PredXF(x.\mathit{next}) = y$.
The predicate says that there the \emph{link} from $x$ to $y$ is reachable.
Proving that exporting this  from the local view of the traversal to the past of the concurrent execution is the key technical contribution of~\cite{PODC10:Hindsight}.

Our previous results allow to establish that $\PReachXK{y}{k}$. Assume that $\seqr$ further reads the field of $y$ and sees that $\PredXF(y)$ holds. Since $\PredXF(y)$ depends on a single location, this means that it holds \emph{now}, at the time of the read.
Our goal is to ensure that $\Past{\QReachXK{y}{k}\land\PredXF(y)}$ is also true, i.e., both
$\QReachXK{y}{k}$  and $\PredXF(y)$ held
at the same time in the past.
The reasoning is as follows:
Let $y.d$ be the field of on which $\PredXF(y)$ depends.
We have that at some point in the past $\QReachXK{y}{k}$ holds.
If $\PredXF(y)$ also held at that time, then we are done.
Otherwise $\PredXF(y)$ was false at the past but it is true now, when we read $y.d$.
Therefore, a write $w$ must have changed $y.d$.
From preservation\footnote{Technically, the important property of a field is that it is reachable iff the object it belongs to is reachable; $\QReachXK{y.d}{k} \Rightarrow \QReachXK{y}{k}$.},
$\QReachXK{y}{k}$ holds also at the time of the write, so after $w$ is performed it holds that $\QReachXK{y}{k} \land \PredXF(y)$.

\clearpage
\section{Necessity of Conditions}\label{Se:NecessityConditios}
In this section we illustrate how reasoning from unsynchronized traversals might be incorrect when the conditions of our framework are not satisfied, disabling its use to reason about concurrent executions from the sequential behavior of the local view. The examples are based on (artificial) modifications of the running example (\Cref{Fi:Running}).

\para{Acyclicity} Consider a traversal of the tree in our running example searching for some key $k$, which starts from some arbitrary node $y$ rather than from $\rootobj$. Assume that the traversal then reaches $\NULL$. The traversal now starts from $\rootobj$, and happens to reach $y$. Based on the local view of the traversal, the operation declares that $k$ is not present in the tree, since $\QReachXK{\NULL}{k}$ holds in the local view (as $\QReachXYK{y}{\NULL}{k}$ and $\QReachXK{y}{k}$ were found).

However, this may not be true for any intermediate state of the concurrent execution: assume that the parent of $y$ when the traversal begins is $x$ and that $x.key = k$. The scenario above is possible if between the first and second phases of the traversal $x,y$ are rotated (see \Cref{Fi:Rotation}), although a node with key $k$ is always present in the tree.

Note that since the modifying procedures are exactly as in the running example, preservation holds in this example. The search order in this example includes edges from $\NULL$ to $\rootobj$, and this is of course not acyclic.  
\TODO{example of temporal acyclicity}

\para{Preservation}
Consider a binary search tree in which rotations are performed in-place, and a traversal exactly as \code{locate} in \Cref{Fi:Running}.
Consider reachability to $x$ when $y,x$ are rotated (see~\Cref{Fi:Rotation}). If the traversal reaches $y$ before the rotation, then the in-place rotation takes place, and the traversal continues, the traversal ``misses'' $x$, and could mistakenly declare that the key of $x$ is not contained in the tree.

Note that the traversal by pointers property is maintained. The acyclicity conditions does not hold (in an insubstantial way, see \Cref{Se:Conc}) but this does not affect the view of the traversal depicted here. (It does affect the traversal if it reaches $x$ instead of $y$ when a rotation takes place.)

\clearpage
\section{Formally Justifying the Validity of Local View Arguments for the Running Example}\label{Se:FormalRunningConditions}


\begin{figure}
\centering
\small

$
\begin{array}{l@{\,}c@{~~}l}
\mathit{Root}(H,H') & \eqdef & \forall o_1.\ H \models \rootobj = o_1 \implies H' \models \rootobj = o_1 \\
\mathit{Key}(H,H') & \eqdef & \forall o_1,k.\ H \models o_1.key = k \implies H' \models o_1.key = k \\
\mathit{Rem}(H,H') & \eqdef & \forall o_1,k.\ H \models o_1.rem \implies H' \models o_1.rem \\[2mm]
\mathit{Acyclic}(H,H') & \eqdef &  \forall o_1,o_2,o_3.\, ((H \models \QReachXK{o_1}{}\land child(o_1,o_2))\land (H' \models child(o_1,o_3)\land o_2\neq o_3)) \\
&&\hspace{4cm} \implies H \models \QReachXYK{o_2}{o_3}{} \vee H\models (new(o_3)\land \neg \QReachXYK{o_3}{o_1}{})
\\
\mathit{Preservation}(H,H') & \eqdef &  \forall o_1.\, H \models \QReachXK{o_1}{k} \implies H' \models \QReachXK{o_1}{k} \vee H'\models o_1.rem \vee H'\models locked(o_1)\land \neg \QReachXYK{root}{o_1}{}  \\
&&\hspace{7mm}H \models \neg \QReachXYK{root}{o_1}{} \implies H'\models locked(o_1) \vee H'\models o_1.rem

\end{array}
$
\caption{\label{Fi:Inv}Transition invariants for the running example. The variables $o_1$, $o_2$, and $o_3$ are interpreted over allocated objects in the heap. The predicate $\QReachXYK{o_1}{o_2}{}$ denotes the fact $o_2$ is reachable from $o_1$ in the heap (by some sequence of accesses to $\mathit{left}$ or $\mathit{right}$). Also, $locked(o_1)$ means that the procedure executing the current step holds a lock on $o_1$, and $new(o_3)$ means that $o_3$ was allocated by the procedure executing the current step and never ``made'' reachable from the $\rootobj$ (i.e., its address was never stored into some $\mathit{left}$ or $\mathit{right}$ field). 
}
\end{figure}

We discuss a particular strategy for proving the conditions of \Cref{Se:FrameworkConditions} above which applies in particular our running example. The acyclicity condition refers to a partial order $\leq_H$ on memory locations in a heap $H$ which in the case of the example is defined by 
\begin{align*}
&\forall o_1,o_2.\ child(o_1,o_2)\implies \forall f\in\{\mathit{keo_2}, \mathit{left}, \mathit{right}, \mathit{del}, \mathit{rem}\}.\ o_1.f\leq_H o_2.f \\
&\forall o_1.\ o_1.\mathit{key} \leq_H \{o_1.\mathit{rem}, o_1.\mathit{del}\} \leq_H \{o_1.\mathit{left}, o_1.\mathit{right}\}
\end{align*}
where $child(o_1,o_2)$ means that $o_2$ is a child of $o_1$. We may write $o_1\leq_H o_2$ to say that all the fields of $o_1$ are before the fields of $o_2$ in $\leq_H$.

The acyclicity condition follows from the \emph{transition invariant} $Acyclic(H,H')$ in \Cref{Fi:Inv}, which describes a relation between the state $H$ before and the state $H'$ after executing any statement in the code of the algorithm (in any concurrent execution). According to this invariant, if some assignment changes the child of some object $o_1$, reachable from the root, from $o_2$ to $o_3$, then $o_3$ was either reachable from $o_2$ (which implies $o_1\leq_H o_3$) or $o_3$ is a ``new'' object which was never reachable from the $root$ and $o_1$ is not reachable from $o_3$ (which implies $o_3\not\leq_H o_1$). In both cases, adding the constraint $o_1$ is ``smaller than'' $o_3$ to the partial order $\leq_H$ will not introduce a cycle. Proving the validity of this invariant is rather easy. The only updates to the children of a node reachable from the $\rootobj$ (when they are not ${\tt null}$) are from \code{removeRight()} at \Cref{Ln:RemoveRight} and \Cref{Ln:RemoveLeft}, and from \code{rotateRightLeft()} at \Cref{Ln:RotateXRight} and \Cref{Ln:RotatePLeft} (for all the other updates, the invariant holds vacuously). Notice that any property of fields of objects which are locked is true in concurrent executions as long as it holds in sequential executions. For instance, $y=z.\mathit{right}$ holds at line \Cref{Ln:RemoveRight} and \Cref{Ln:RemoveLeft}, and it implies on its own the invariant. A similar reasoning can be done at \Cref{Ln:RotatePLeft}. For \Cref{Ln:RotateXRight}, the new child $z$ of $x$ is a newly allocated object whose children are not on a path to $x$.

The transition invariant $Preservation(H,H')$ in \Cref{Fi:Inv} implies a property which is even stronger than the preservation condition: for every execution $e$ of the concurrent algorithm, and every update $w$ in $e$ to a heap object $o$, if $\QReachXK{o}{k}$ became true at some moment before $w$, then it remains true until $w$ gets executed (there is no requirement that $w$ overlaps in time with a read-only code fragment).
According to $Preservation(H,H')$, every other write $w'$ that happens after the moment when $\QReachXK{o}{k}$ became true will either maintain the validity of this predicate, or it will turn it to false, but then it will either hold a lock on $o$ or set its field $\mathit{rem}$ to $\true$ (and according to $Rem(H,H')$ the field $\mathit{rem}$ never changes from $\true$ to $\false$). The latter two cases are impossible since $w$ is enabled only if it holds a lock on $o$ and the field $\mathit{rem}$ is $\false$.
 Concerning the proof of $Preservation(H,H')$, since the field $\mathit{key}$ of every object, and the variable ${\tt root}$ are immutable (stated formally in $Key(H,H')$ and $Root(H,H')$), the only way to modify the validity of a predicate $\QReachXK{o}{k}$ is by changing the pointer fields $\mathit{left}$ or $\mathit{right}$. The only such updates occur in the procedures \code{removeRight()} and \code{rotateRightLeft()}. These updates can modify search paths only by removing or inserting keys, the updates at \Cref{Ln:RemoveRight},  \Cref{Ln:RemoveLeft}, and \Cref{Ln:RotatePLeft} remove the key of {\tt y} while the update at \Cref{Ln:RotateXRight} inserts the key of {\tt y} (or {\tt w}). The interesting case is when keys are removed from search paths: a predicate $\QReachXK{o}{k}$ can become $\false$, but only if $o$ contains the removed key and $o$ becomes unreachable from the $\rootobj$. Also, this can happen only if the procedure executing the current step holds a lock on $o$ or if the field $\mathit{rem}$ is set to $\true$
 (an unlock validates $Preservation(H,H')$ since it can happen only when the field $\mathit{rem}$ is already set).

\clearpage







\lstset{language={program},style=lnumbers,firstnumber=last}

\begin{figure}
\centering
\begin{tabular}{p{12cm}}
\begin{lstlisting}
type LFN 
  immutable int key 
  LFN$\times$bool next=$\langle$ref,mark$\rangle$

LFN tail$\leftarrow$new($+\infty$,null)
LFN head$\leftarrow$new($-\infty$,tail)
\end{lstlisting}
\begin{lstlisting}
bool add(int key) 
  LFN newNode, pred, succ

  (pred,succ)$\leftarrow$find(key) $\label{ln:lfn:add-check-key}$
  if (succ.key = k)
    return false 

  newNode$\leftarrow$new SNL(key,(succ,false)) 
  bool added$\leftarrow$CAS(&pred.next,(succ,false),(newNode,false))  $\label{ln:lfn:add-cas}$   
  if ($\neg$added)
    restart

  return true   
\end{lstlisting}
\begin{lstlisting}
bool remove(int key) 
  LFN newNode, pred, succ 

  (pred,succ)$\leftarrow$find(key)
  if (succ.key $\neq$ k)  $\label{ln:lfn:remove-check-key}$ 
    return false  $\label{ln:lfn:remove-check-key-ret}$ 

  LFN nodeToRemove$\leftarrow$succ
  (succ,bool)$\leftarrow$nodeToRemove.next

  while (true)
    bool iMarkedIt$\leftarrow$CAS(&nodeToRemove.next,(succ,false),(succ,true)) $\label{ln:lfn:remove-cas}$
    (succ,marked)$\leftarrow$nodeToRemove.next
    if (iMarkedIt)
      find(key)
      return true
    if (marked)
      return false $\label{ln:lfn:remove-marked-ret}$
\end{lstlisting}
\begin{lstlisting}
LFN$\times$LFN find(int key) 
  bool find, snip, marked, cont$\leftarrow$true   
  LFN pred$\leftarrow$null, curr$\leftarrow$null, succ$\leftarrow$null
  while (cont)
    pred$\leftarrow$head
    curr$\leftarrow$pred.next.ref
    while (cont)
      (succ,marked)$\leftarrow$curr.next $\label{ln:lfn:find-get-next}$
      while (marked)               $\label{ln:lfn:find-while-mark}$
        snip$\leftarrow$CAS(&pred.next,(curr,false),(succ,false)) $\label{ln:lfn:find-snip}$ $\label{ln:lfn:find-cas}$
        if ($\neg$snip) 
          restart
        curr$\leftarrow$succ  
        (succ,marked)$\leftarrow$curr.next $\label{ln:lfn:find-get-next-in-while}$
      if (curr.key < key) $\label{ln:lfn:find-check-key}$
         pred$\leftarrow$curr $\label{ln:lfn:find-set-pred}$
         curr$\leftarrow$succ $\label{ln:lfn:find-set-curr}$
      else
         cont$\leftarrow$false     
  return (pred,curr)  
\end{lstlisting}
\begin{lstlisting}
bool contains(int key) 
  LFN curr, succ
  bool marked$\leftarrow$false

  curr$\leftarrow$head
  succ$\leftarrow$pred.next.ref
  while (curr.key $<$ key) $\label{ln:lfl:contains:while}$
    $\{\Past{\lflQReachXK{curr}{key} \land (curr.mark \iff marked)} \}$ $\label{ln:lfl:contains:inv}$
    curr$\leftarrow$succ
    (succ,marked)$\leftarrow$curr.next

  $\{\Past{\lflQReachXK{curr}{key}  \land  (curr.mark \iff marked)}  \land key \leq curr.key \}$
  return (curr.key$=$key $\land$ $\neg$marked)
\end{lstlisting}

\end{tabular} 
\caption{\label{Fi:LFList}
Lock-free concurrent list~\cite[Chapter 9.8]{TAOMPP}
} 
\end{figure}

\section{Example: Lock-Free List-Based Concurrent Set}\label{Se:LockFreeList}

In this section, we apply our approach to verify the lock-free list-based concurrent set algorithm shown in~\Cref{Fi:LFList}.   
The code of the algorithm is based on the algorithm of~\cite[Chapter 9.8]{TAOMPP},  adapted to our language.
The algorithm is explained in detail in~\cite{TAOMPP}.
Thus, we only describe the parts necessary to understand our linearizbility proof and the assertions,  written inside curly braces, which annotate the code of the \code{contain} procedure.

The set algorithm uses an underlying sorted
linked-list of dynamically-allocated objects of type \code{LFN}, which we refer to as \emph{nodes}.
Every node has three fields: an immutable integer field {\tk} storing the
key of the node, a pointer field {\tref} pointing to a successor node 
(or to a designated \code{null} value),
and a boolean field {\tm} indicating
that the node was logically deleted from the list. 

The {\tref} and {\tm} fields of a node can be accessed atomically: we encapsulate  these fields inside a pair field {\tmnext}  comprised of a reference and a boolean field which can be accessed atomically. 
We write \code{x.next.ref} and \code{x.next.mark} to denote accessing the {\tn} and {\tm} fields of the node pointed to by \code{x} separately. We use the notation \code{(succ,marked)}$\leftarrow$\code{x.next},
where \code{succ} is a pointer variable and \code{marked} a boolean variable, 
to denote an atomic assignment of \code{x.next.ref} and \code{x.next.mark} to \code{succ} and \code{marked},
respectively. 
Similarly, we write \code{x.next}$\leftarrow$\code{(succ,marked)} to denote an atomic assignment to the two components of the \code{next} pair of fields of the node pointed to by \code{x}. 
We write \code{CAS(\&x.next,(currref,currmark),(newref,newmark))} to denote a \emph{compare-and-set} operation which atomically sets the contents of the {\tmnext} field of \code{x} to
\code{(newref,newmark)}, provided that its current value is \code{(newref,newmark)}.
When the \tm-component of the \code{next} field of a node is set, we say that the field, as well as the node itself, are  \emph{marked},
otherwise, we say that they are \emph{unmarked}.

\begin{remarknum}
The original code~\cite[Figures 9.24 to 9.27]{TAOMPP} is written in Java and keep the {\tref} and {\tm} fields of node using a \emph{markable pointer} (see~\cite[Pragma 9.8.1]{TAOMPP}) to allow reading, writing, and applying CAS to the {\tn} and {\tm} fields
simultaneously.
\end{remarknum}


The  list has designated sentinel \emph{head} and \emph{tail}
nodes. The \emph{head} node is always pointed to by the shared
variable \code{head} and contains the  value $-\infty$. The
\emph{tail} node is always pointed to by the shared variable
\code{tail}, and contains the  value $\infty$. The value $-\infty$
(resp. $\infty$) is smaller (resp. greater) than any possible value
of a key. When the algorithm starts, it first sets \emph{tail} to be the successor of \emph{head} and sets \emph{tail}'s successor to be \code{null}. The sentinel nodes remain unmarked throughout the execution.

The set algorithm is comprised of three interface procedures: \code{add},
\code{remove}, and \code{contains}.
The first two
use the internal \code{find} procedure to traverse the list
and prune out marked nodes:
when \code{find} is invoked to locate a key \code{key},
it traverses the list starting from the \emph{head} node until it reaches an unmarked node with a key greater than \code{key}. During the traversal it removes marked nodes (\Cref{ln:lfn:find-snip}). If an attempt to remove a node fails, the procedure restarts.
In contrast, the \code{contains} method's  traversal of the list is \emph{optimistic}: it is done
without any form of synchronization. As a result,  while a
thread is traversing the list, other threads might concurrently
change the list's structure.
When verifying this algorithm,  our approach helps in proving that \texttt{contains} is linearizable,
which is the most difficult part of the proof.
Proving the linearizability of \code{add} and \code{remove} can be done using a  rather standard   invariant-based concurrent reasoning 
as discussed below.

\paragraph*{Verifying assertions using invariants-based concurrent reasoning}
Procedures \texttt{add}, \texttt{remove}, and \texttt{find}  maintain several  \emph{state} invariants:\footnote{Recall that the \code{contains} does not modify the shared state.}
\begin{compactitem}[]
\item 
($\lflI{rT}$) the tail node is always reachable from the head node, where reachability between nodes is determined in this section by following  $\tref$ fields;
\item 
($\lflI{UB}$) all unmarked nodes are reachable from the \emph{head} node; 
and 
\item 
($\lflI{<}$) if node $v$ is the \tref-successor of node $u$ then the key of $v$ is strictly greater than that of~$u$. 
\end{compactitem}
In addition, the procedures maintains the following   \emph{transition} invariants:
\begin{compactitem}[]
\item 
($\lflTI{k}$) the \code{key} node is immutable;  
\item
($\lflTI{mn}$) the \code{next}    fields becomes immutable once it gets marked;
and
\item
($\lflTI{mr}$) right after a node gets marked, it is reachable from the head.
\end{compactitem}
Verifying the invariants hold is rather straightforward as it merely requires local reasoning about each mutation. 
For example, it is easy to see that invariant $\lflTI{mn}$ holds:
The modifications (\Cref{ln:lfn:add-cas,ln:lfn:remove-cas,ln:lfn:find-cas}) are   done using a \code{CAS} operation which may succeed only if the modified \code{next} pair of fields is unmarked.
Furthermore, a node $u$ gets marked only in \Cref{ln:lfn:remove-cas}. Hence, $u$ is reachable at that time. Note that the compare-and-set command cannot affect the \tref-field of $u$'s predecessor.

\paragraph*{Verifying linearizability}
We prove the linearizability of the algorithm  using an \emph{abstraction function}
$\repfunc: H \totalto \powerset(\mathbb{N})$ that maps a concrete memory state of the list to  the {\em abstract set} represented by this state, and showing that \code{add}  and \code{delete} manipulate this abstraction according to their specification and that \code{find} does not modify it.
We define $\repfunc$ to map $H$ to the set of keys of the unmarked nodes.
Note that by invariants $\lflI{rT}$, $\lflI{UB}$, and  $\lflI{<}$, these nodes 
are part of the sorted list segment connecting the \emph{head} and \emph{tail} nodes.
We refer to this list segment as the \emph{backbone list}. 

\subparagraph*{Verifying linearizability of add and remove}
The proof that  \code{add} and \code{remove}   are linearizable 
follows directly from the invariants once we establish the following properties 
of \code{find}:
(a) it does not change the abstract set represented by the list, and
(b) the  pointers (\code{pred},\code{curr}) it returns point to nodes \emph{pred} and \code{curr}, respectively, such that
(i)
the key of  \emph{pred} (resp. \emph{curr}) is smaller than (resp. greater or equal to) \code{key},
(ii)
\emph{pred} was unmarked and the \tref-predecessor of  \emph{curr} at some point during the traversal, 
and  
(iii) 
at some (perhaps different) time point during the traversal,  \emph{curr} was unmarked.

To verify property (a), we  observe that  \code{find}  removes the marked node pointed to by \code{curr} by redirecting the \code{ref}-field of its predecessor (pointed to by \code{pred}) to point to \code{succ}---\code{curr}'s \code{ref}-successor:
Using compare-and-set  (\Cref{ln:lfn:find-snip}) ensures that the removal succeeds only if \code{pred} is unmarked and its \tref-field points to \code{curr}.
The fact that \code{succ} is the \tref-successor of \code{curr} is ensured by transition invariant $\lflTI{mn}$ which prohibits the modification of marked node. 
Property (a) holds because cutting out marked nodes this way does not affect the reachability of unmarked nodes from the head.

Property (b.i) follows from the check made in \Cref{ln:lfn:find-check-key} and the immutability of keys.
To verify properties (b.ii) and (b.iii), it suffices to observe that 
after \code{find} traverses the \code{next.ref}-field of the node \emph{curr} pointed to by \code{curr} (\Cref{ln:lfn:find-get-next,ln:lfn:find-get-next-in-while}) it ensures that \emph{curr} is unmarked (\Cref{ln:lfn:find-while-mark}) before it updates \code{pred} and \code{curr} 
(\Cref{ln:lfn:find-set-pred,ln:lfn:find-set-curr}).\footnote{Recall that the {\tref} and {\tm} fields are read in one atomic action.}

The linearizability of  invocations of \code{add} and \code{remove} which return \code{true}  follows from invariant  $\lflI{UB}$ which ensures that as they modify unmarked nodes, these nodes must be reachable from the head. This, together with property (a), shows that  adding a new node or marking an existing unmarked one affects the represented set in the intended way. 

The linearizability of   invocations of
\code{add}   which return \code{false}
follows from property (b.iii) and the check made in \Cref{ln:lfn:add-check-key}:
The former ensures that the node pointed to by \code{curr} was unmarked during the traversal of \code{find} and the latter that the key of that node is the one the procedure attempts to add.

The linearizability of \emph{unsuccessful} invocations of \code{remove}, i.e.,
ones which return \code{false}, 
can justified using two different reasons:
\begin{compactitem}
\item 
\code{remove} returning \code{false} in \Cref{ln:lfn:remove-check-key-ret} can be justified by properties (a) and (b.ii), which, together, ensure that there was some point during the execution of \code{find} in which the node \emph{pred} pointed to by \code{pred} was reachable from the head and its successor was the node \emph{curr} pointed to by \code{curr}. Hence, the latter was reachable too. As the key $k$ the procedure tries to remove is bigger than the key of \emph{pred} and smaller than the key of \emph{curr}, the  sortedness of the list ensures  that at this time $k$ was not the key of any unmarked node.

\item 
\code{remove} returning \code{false} in \Cref{ln:lfn:remove-marked-ret} can be justified by property (a),  the check made in \Cref{ln:lfn:add-check-key}, and property   (b.iii):
The first two ensure that the node pointed to by \code{curr} was unmarked during the traversal of \code{find} and the last one was that at a later point this node was marked. By transition invariant $\lflTI{mr}$, it follows that at some time point during the execution of \code{remove} a marked node with the key removes attempts to delete was reachable from the head, and as the list is sorted, that key was not part of the set the list represents.
\end{compactitem}





\subparagraph*{Verifying linearizability of  contains}
We use our framework to verify the linearizability of \code{contains} by
defining the notion of an order over memory locations, the notion of \emph{valid search path} for
key $k$ that starts at the \emph{head} node, and proving that the code satisfies the acyclically and preservation conditions.

As in our running example, we define the order over memory locations based on reachability. We say that there is a valid search path to an object $o_x$ for key $k$ from the head of the list, denoted by
$\lflQReachXK{o_x}{k}$,
if  $o_x$ is reachable from the head node to $o_x$ and its key is smaller than $k$. Formally, search paths are defined as follows:
$$
\begin{array}{l}
\lflQReachXYK{o_r}{o_x}{k}
\eqdef
\exists o_0,\ldots,o_m.\,
o_0=o_r \land o_m=o_x \land {}
\\[2pt]
\qquad
\forall i=1..m.\, o_{i-1}.\tmnext.\tref=o_i \land  o_{i-1}.\code{key} < k \ .  

\end{array}
$$

The acyclicity of the order stems from the immutability of keys and  invariant $I_{<}$ which ensures that cycles are impossible.

To prove the preservation of search paths to locations of modification it suffices to note that as marked nodes are never modified, it suffices to show the property hold for unmarked ones.
Note that neither marking a nor changing the successor of a node affects the search paths which go   \emph{through} it:
Invariant $\lflI{<}$ ensures that adding a node $v$ in between nodes $u$ and $w$ does not break any search paths which goes through $w$: These must be for  keys greater than that of $w$, and hence of $v$.
Removing a marked node may merely shorten a search path to an unmarked node.
Marking a node has not effect of search paths.

To verify the linearizability of \code{contains} we
establish the loop invariant
$$
\lflQReachXK{curr}{key} \land (curr.mark \iff marked)
$$
in \Cref{ln:lfl:contains:inv} using sequential reasoning. 
This is straightforward.
We then lift it to concurrent executions (using its past form) by applying
the extension of our framework discussed in \Cref{rem:ReachWithField}.
This invariant, together with the definition of a search path ensures that as we get out of the loop only if we reached an unmarked node with the key that we look for, and in this case \code{marked} is false, or  
that during the traversal we never encountered that key that we look for or that this key was in a marked node. In either cases, the return value correspond to the contents of the abstract set at that time.



\clearpage




\lstset{language={program},style=lnumbers,firstnumber=last}

\begin{figure}
\centering
\begin{tabular}{p{8cm}p{6.7cm}}
\begin{lstlisting}
type SLN 
  immutable int key 
  immutable int topLevel 
  SLN $\times$ bool next[L]

SLN tail$\leftarrow$new($+\infty$,L,(null,false),...,(null,false))
SLN head$\leftarrow$new($-\infty$,L,(tail,false),...,(tail,false))


\end{lstlisting}
\end{tabular} 
\caption{\label{Fi:LFSkipList:Type}
Lock-free concurrent skiplist: Type declaration and the \code{head} and \code{tail} global variables pointing to the first and last, respectively, sentential nodes of the list. (See~\cite[Fig 14.10]{TAOMPP}).  
} 
\end{figure}


\section{Example: Lock-Free Skiplist-Based Concurrent Set}\label{Se:LockFreeSkipList}
 
In this section, we apply our approach to verify the lock-free skiplist-based concurrent set algorithm shown in~\Cref{Fi:LFSkipList:Type,Fi:LFSkipList:Contains,Fi:LFSkipList:Find,Fi:LFSkipList:Add,Fi:LFSkipList:Remove}.   
The code of the algorithm is based on the algorithm of~\cite[Chapter 14.4]{TAOMPP},  adapted to our language.
The algorithm is explained in detail in~\cite{TAOMPP}.
Thus, we only describe the parts necessary to understand our linearizbility proof and the assertions,  written inside curly braces, which annotate the code of the \code{contain} procedure (\Cref{Fi:LFSkipList:Contains}).

The set algorithm uses an underlying concurrent skiplist comprised of  
dynamically-allocated objects of type \code{SLN} (see \Cref{Fi:LFSkipList:Type}), which we refer to as \emph{nodes}.
Every node has three fields: an immutable integer field {\tk} storing the
key of the node, 
an array {\tmnext}  with $L$ entries, where each entry contains a pair comprised of a pointer field {\tref} and a 
boolean field {\tm} which allows to link every node in multiple levels,
and an integer field {\ttl} which determines the number of populated entries.
We refer to the list obtained by the links at the $i$th entry of the nodes array of links as the list at level $i$.
Roughly speaking, every node $u$ is part of lock-free lists (see \Cref{Se:LockFreeList})
at levels $0..u.\ttl$. 
The bottom list (\tmnext[0]) is the \emph{main} list, and every list $i$,
where $0<i\leq u.\ttl$, serves as a shortcut which allows to bypass multiple nodes of the list at level $i-1$.

The  skiplist has designated sentinel \emph{head} and \emph{tail}
nodes. The \emph{head} node is always pointed to by the shared
variable \code{head} and contains the  value $-\infty$. The
\emph{tail} node is always pointed to by the shared variable
\code{tail}, and contains the  value $\infty$. The value $-\infty$
(resp. $\infty$) is smaller (resp. greater) than any possible value
of a key. When the algorithm starts, it first sets \emph{tail} to be the successor of \emph{head} in all $L$ levels, and initializes all of \emph{tail}'s \tref-fields to be \code{null}. The sentinel nodes remain unmarked throughout the execution.

The set algorithm is comprised of three interface procedures: \code{add},
\code{remove}, and \code{contains}. 
The first two
use the internal \code{find} procedure to traverse the list
and prune out marked nodes.
In contrast, the \code{contains} method's  traversal of the list is \emph{optimistic}: it is done
without any form of synchronization. As a result,  while a
thread is traversing the list, other threads might concurrently
change the list's structure. 
When verifying this algorithm,  our approach helps in proving that \texttt{contains} is linearizable, 
which is, as in the case of the lock-free list, the most difficult parts of the proof. 
Proving the linearizability of \code{add} and \code{remove} follows rather easily using invariant-based concurrent reasoning 
as discussed below.

\paragraph*{Verifying invariants using concurrent reasoning}
Procedures \code{add}, \code{remove}, and \code{find}  maintain several  \emph{state} invariants.\footnote{Recall that the \code{contains} does not modify the shared state.} In particular, for every $i=0..L$, the  list at level $i$ maintain all the state and transition invariants of the lock-free list-based set algorithm (see \Cref{Se:LockFreeList}).
In addition, the skip list maintain the following state
invariants:
\begin{compactitem}[]
\item 
($\slI{mu}$) if the ${\tmnext}[i]$ field of a node $u$ 
is marked then so are all the fields $u.{\tmnext}[j]$ for $i < j \leq L$.
($\slI{sub}$) if node $u$ precedes node $v$ at level $i$ and both nodes are unmarked at level $i$at level $i$ then $u$ precedes $v$ at  level $j$ for any $0 \leq j < i$.
\end{compactitem}

Verifying the invariants hold is rather straightforward as it merely requires local reasoning about each mutation. 
For example, it is easy to see that the \code{next} field of a marked node is never modified:
The modifications (\Cref{ln:sl:add-cas,ln:sl:add-setnext-cas,ln:sl:add-main-cas,ln:sl:remove-cas,ln:sl:remove-main-cas,ln:sl:find-cas}) are   done using a \code{CAS} operation which may succeed only if the modified \code{next} pair of fields is unmarked. 
To verify invariant $\slI{mu}$, we only need to observe that \code{remove} marks the entries of the {\tmnext} array from top to bottom (\Cref{ln:sl:remove-main-cas}). Note that as a marked \tmnext-field never gets modified,   invariant $\slI{mu}$ holds even if a thread tries to remove a node which is still being added to the list.
To verify invariant $\slI{sub}$, we   first observe that \code{add} links a new node in (level-wise) a 
bottom up fashion (\Cref{ln:sl:add-main-cas}). Thus, a node $v$ gets linked to the $i$th level, for any $0\leq i < L$ before it gets linked in level $i+1$. We then apply invariant $\slI{mu}$ to realize that if nodes $u$ and $v$ are unmarked at level $i$ then they are unmarked at level $i+1$.
By invariant $\lflI{UB}$ of the lock free list (see \Cref{Se:LockFreeList}), this means that the node is reachable at the lists at level $i+1$ and at level $i$.
By $\lflI{<}$, the lists at all levels are sorted.  Thus, if $u$ precedes $v$ at level $i$ it precedes is at level $i-1$ too.

\paragraph*{Verifying linearizability}
We prove the linearizability of the algorithm  using an \emph{abstraction function}
$\repfunc: H \totalto \powerset(\mathbb{N})$ that maps a concrete memory state of the list to  the {\em abstract set} represented by this state, and showing that \code{add}  and \code{delete} manipulate this abstraction according to their specification and that .\code{find} does not modify it.
We define $\repfunc$ to map $H$ to the set of keys of the unmarked nodes of the main (bottom) list.
Note that by invariants $\lflI{rT}$, $\lflI{UB}$, and  $\lflI{<}$), these nodes 
are part of the \emph{backbone list} of the main list---the sorted list connecting the \emph{head} and \emph{tail} nodes.

\subparagraph*{Verifying linearizability of  add and
remove}

The proof that  \code{add} or \code{remove}, shown in 
\Cref{Fi:LFSkipList:Add,Fi:LFSkipList:Remove},
are linearizable follows almost immediately from the invariants once we establish the following properties of \code{find}, shown in
\Cref{Fi:LFSkipList:Find}:
(A) \code{find} populates the pair of input arrays with pointers to predecessors and successors of the searched key at every level, and (B) it returns \code{true} if and only if it found an unmarked node at the bottom list containing the searched node. Furthermore, in this case, it sets \code{succs[0]} to point to this node.

To prove property (A), we note that, roughly speaking, the \code{find} procedure of the skiplist traverses the lists at all the levels starting from the top lists and making its way down to the bottom list.
It fills the $i$th entry of the \code{preds} and \code{succs} arrays, for every $0 \leq i \leq L$, at with pointers to the predecessor and successor nodes of the search key at the list at level $i$.
As in the \code{find} procedure of the lock-free list, it prunes out marked nodes as it goes over the lists (\Cref{ln:sl:find-snip}).
The most tricky aspect of \code{find} is that its traversal at level $i-1$, for $0 < i \leq L$,
does not start from the head of the list but from the predecessor  node of the searched key at level $i$.  \code{find} ensures that it does not miss a  node $u$ containing the desired key which is unmarked at the bottom level by switching the traversal from level $i$ to level $i-1$ in a 
node \emph{pred} which, being a predecessor node, has a smaller key than the one \code{find} searches for and 
is  unmarked at level $i-1$ (\Cref{ln:sl:find-start,ln:sl:find-check-start}).
By invariant  $\slI{sub}$, \emph{pred} is unmarked at all levels $0..i-1$, by invariant  $\lflI{UB}$,
$u$ is on the backbone of the list at   levels $0..i-1$, and by invariant $\lflI{<}$, $u$  appears after \emph{pred} at these lists.

To prove property (B), we observe that \code{find} gets out of the \code{for}-loop only after it set \code{succs[0]} to point to a node with a key greater or equal to the searched key 
(\Cref{ln:sl:find-check-key,ln:sl:find-set-succs})
which was unmarked  during the traversal (\Cref{ln:sl:find-check-mark-curr}).


\lstset{language={program},style=lnumbers,firstnumber=last}

\begin{figure}
\centering
\begin{tabular}{p{10cm}}
\begin{lstlisting}
bool find(int key, SLN preds[L], SLN succs[L]) 
  bool find, snip, marked, cont   
  SLN pred$\leftarrow$null, curr$\leftarrow$null, succ$\leftarrow$null
  pred$\leftarrow$head
  for (int level$\leftarrow$L; 0 $\leq$ level; level--) $\label{ln:sl:find-snip}$
    (curr,marked)$\leftarrow$pred.next[level] $\label{ln:sl:find-start}$
    if (marked)                               $\label{ln:sl:find-check-start}$
      restart
    cont$\leftarrow$true 
    while (cont)
      (succ,marked)$\leftarrow$curr.next[level] 
      while (marked)                          $\label{ln:sl:find-check-mark-curr}$
        snip$\leftarrow$CAS(&pred.next[level],(curr,false),(succ,false)) $\label{ln:sl:find-cas}$
        if ($\neg$snip) 
          restart
        (succ,marked)$\leftarrow$curr.next[level]
      if (curr.key < key)     $\label{ln:sl:find-check-key}$
        pred$\leftarrow$curr  $\label{ln:sl:find-set-pred}$
        curr$\leftarrow$succ
      else
        cont$\leftarrow$false     
    preds[level]$\leftarrow$pred. 
    succs[level]$\leftarrow$succ $\label{ln:sl:find-set-succs}$
  return curr.key=key   
\end{lstlisting}
\end{tabular}
\caption{\label{Fi:LFSkipList:Find}
Lock-free concurrent skiplist algorithm: Procedure \code{find}. (See~\cite[Fig 14.13]{TAOMPP}).
} 
\end{figure}

\begin{remark}
We modified the  \code{add} procedure  of~\cite[Fig 14.11]{TAOMPP}  to update the successors of a newly added node using a compare-and-swap (\Cref{ln:sl:add-setnext-cas}) because the original version has a subtle race between concurrent \code{add} and \code{remove} of the same node.
Also, following~\cite{phd:Fraser04}, we added to procedure \code{find}  a check that the node \code{pred} in which the traversal switches to a lower level (\Cref{ln:sl:find-start})   is unmarked at the new level (\Cref{ln:sl:find-check-start}). This simplified the proof of \code{find}, which is   done outside our framework.
The \code{find}  procedure of~\cite[Fig 14.13]{TAOMPP} does not make this check, which is indeed unnecessary: As \code{pred} was checked to be  unmarked at the previous, higher, level, invariant $\slI{mu}$ ensures that at the time it was unmarked at the new, lower, level. 
\end{remark}

The proof of linearizability of \code{add} and \code{remove} is carried out essentially in the same way it is done for the lock-free list (see \Cref{Se:LockFreeList}) when applied to the main list. It is possible to do so because property (A) ensures that the pair of pointers \code{(preds[0],succs[0])} returned by the skiplist's \code{find} fulfills the same conditions as the pair \code{(pred,curr)}  returned by the lock-free list's \code{find}. (See property (b) in \Cref{Se:LockFreeList}).


\lstset{language={program},style=lnumbers,firstnumber=last}

\begin{figure}
\centering
\begin{tabular}{p{12cm}}
\begin{lstlisting}
bool add(int key) 
  int topLevel$\leftarrow$random(0..L)
  SLN[] preds$\leftarrow$ new SLN[L] 
  SLN[] succs$\leftarrow$ new SLN[L] 
  SLN newNode, succ, newSucc

  bool found$\leftarrow$find(key,preds,succs)
  if (found)
    return false 
     
  newNode$\leftarrow$new SNL(key,topLevel,null,...,null) 
  (succ,marked)$\leftarrow$succs[0]
  newNode.next[0]$\leftarrow$(succ,marked)
  bool setnext$\leftarrow$CAS(&pred.next[0],(succ,false),(newNode,false))    $\label{ln:sl:add-cas}$
  if ($\neg$setnext)
    restart

  for (int level$\leftarrow$1; level $\leq$ L; level++)
    bool linked$\leftarrow$false
    while ($\neg$linked)
      succ$\leftarrow$succs[level].ref
      newSucc$\leftarrow$newNode.next[level].ref
      setnext$\leftarrow$CAS(&newNode.next[level],(newSucc,false),(succ,false)) $\label{ln:sl:add-setnext-cas}$  
      if ($\neg$setnext)
        return true
      linked$\leftarrow$CAS(&preds[level].next,(newNode,false),(succ,false))   $\label{ln:sl:add-main-cas}$
      if ($\neg$linked)
        bool newMark$\leftarrow$newNode.next[level].mark 
        if (newMark)
          return true
        find(key,pred,succ)
   return true   
\end{lstlisting}
\end{tabular}
\caption{\label{Fi:LFSkipList:Add}
Lock-free concurrent skiplist algorithm: Procedure \code{add}. (See~\cite[Fig 14.11]{TAOMPP}).
} 
\end{figure}



\lstset{language={program},style=lnumbers,firstnumber=last}

\begin{figure}
\centering
\begin{tabular}{p{12cm}}
\begin{lstlisting}
bool remove(int key) 
  SLN[] preds$\leftarrow$ new SLN[L] 
  SLN[] succs$\leftarrow$ new SLN[L] 
  SNL succ  
  bool marked, found     
  int level

  found$\leftarrow$find(key,preds,succs)
  if ($\neg$found)
    return false 
     
  SLN nodeToRemove$\leftarrow$succs[0] 
  for (level $\leftarrow$nodeToRemove.topLevel; 1 $\leq$ level; level--)
    (succ,marked)$\leftarrow$nodeToRemove.next[level] 
    while ($\neg$marked)
      CAS(&nodeToRemove.next[level],(succ,false),(succ,true))     $\label{ln:sl:remove-cas}$      
      (succ,marked)$\leftarrow$nodeToRemove.next[level] 
    
  while (true)
    bool iMarkedIt$\leftarrow$CAS(&nodeToRemove.next[level],(succ,false),(succ,true)) $\label{ln:sl:remove-main-cas}$
    (succ,marked)$\leftarrow$nodeToRemove.next[0]
    if (iMarkedIt)
      find(key,preds,succs)
      return true
    if (marked)
      return false


\end{lstlisting}
\end{tabular}
\caption{\label{Fi:LFSkipList:Remove}
Lock-free concurrent skiplist algorithm: Procedure \code{remove}. (See~\cite[Fig 14.12]{TAOMPP}).
} 
\end{figure}

\subparagraph*{Verifying linearizability of contains}
We use our framework to verify the linearizability of \code{contains}, shown in \Cref{Fi:LFSkipList:Contains}, by
defining the notion of an order over memory locations, the notion of \emph{valid search path} for
key $k$ that starts at the top level list entry of the \emph{head} node, and proving that the code respect the acyclically and preservation conditions. 

Recall that the order is defined for memory locations, i.e., at the granularity of fields.
We define the order over memory locations containing entries of the \code{next} arrays of nodes in the following way: All locations pertaining to entries at level $i$ are smaller than the ones at level $j$ for any $i<j$.
We define the order between locations pertaining to entries at the same level according to reachability, as we did in the case of the lock-free list (see \Cref{Se:LockFreeList}).   
The other fields are immutable, and thus their order is immaterial. In particular, it is easy to modify the code so that the key field of a node is read only once in a traversal when moving between levels; this is of course equivalent.

We say that there is a \emph{valid search path} to a location to  the $i$th entry of the {\tmnext} array of the node pointed to by $x$, denoted by $\slQReachXK{x.\inext_i}{k}$,
if it is possible to reach from the top-level \code{next}-field of the  node $\slheadobj$ to
that entry by either (i) traversing over links at the same level if they originate from nodes which are either marked at that level or that their key is smaller than $k$ or (ii) descend to a lower level entry in a node whose key is not greater than $k$ and that it is unmarked at the current level.
Formally, search paths   are defined as follows:
$$
\begin{array}{l}
\QReachXYK{o_r,i}{o_x,j}{k}
\eqdef
\exists o_0,\iota_0,\ldots,o_m,\iota_m.\,
o_0=o_r \land \iota_0=i \land o_m=o_x \land \iota_m=j \land {}\\
\qquad
\forall i=1..m.\, \slnextnode(o_{i-1},\iota_{i-1},k,o_i,\iota_i)\,, \text{ and }
\\[2pt]
\qquad
\qquad
\slnextnode(o_{i-1},\iota_{i-1},k,o_i,\iota_i) = {}\\
\qquad\qquad\qquad	
	(o_{i-1}.\inext_{i-1} = o_i \land (o_{i-1}.\imark \lor o_{i-1}.\ikey < k)) 
	\lor   {}\\
\hfill
	(o_{i-1} = o_i \land (\neg o_{i-1}.\imark \land o_{i-1}.\ikey \leq k) \ .
\end{array}
$$
Note that a search path to $k$ might go at the bottom level through nodes with a key greater than $k$, but these nodes must be marked.
Also note that if $\slQReachXK{x.\inext_0}{k}$ holds and $x.\ikey>k$ this indicates that $k$ is not in the abstract set represented by the list
(because a valid search path to $k$ does not continue past a node with key $k$ and all the list elements are linked in the lowest level).

The acyclicity of the order stems from the immutability of keys, the fact that the order between entries at different levels never changes, and invariant $I_{<}$ of the lock-free list which ensures that cycles are impossible between entries at the same level.

To prove the preservation of search paths to locations of modification it suffices to note that as marked nodes are never modified, it suffices to show the property hold for unmarked ones.
Note that neither adding a link nor removing one changes the search paths which goes \emph{through} unmarked {\tmnext} fields:
Only marked \tmnext-fields are removed (\Cref{ln:sl:find-snip,ln:sl:remove-cas,ln:sl:remove-main-cas}).
As no search path goes from these nodes directly to a lower level entry, this change may only shorten same-level search paths.
Adding a \tmnext-field into a list at level $l$ (\Cref{ln:sl:add-cas,ln:sl:add-main-cas}) does not break a search path that used to go  through 
its predecessor because: the predecessor field remains unmarked (thus there is no effect of search paths that goes down a level) and, 
as discuss in \Cref{Se:LockFreeList}, the list at level $l$ remains sorted, and thus the key of the node is smaller than all the keys in the following nodes.
Marking a node  $v$ adds search paths which go through the marked links, but does not remove any.
However, it may remove search paths that used to   switch at $v$ to a lower level.
Luckily, as the \emph{head} node is never marked, invariant $\slI{sub}$ ensures that for any such search path which was removed there is another  valid search path which goes through an unmarked level-$l$-predecessor $u$ of $v$ that gets to the $l-1$ {\tmnext} entry of $v$ by 
going down a level at $u$ and going through $l$-level links to get to the location of $u.\tmnext[l-1]$'s entry.


To verify the linearizability of \code{contains},
we
establish three loop invariants.
(Using sequential reasoning, establishing the present   form of these invariants is straightforward.)
\begin{compactitem}
\item
The outer (\code{for}) loop invariant ensures that in \Cref{ln:sl:contains:inv:for:pred},
when the procedures starts traversing a new level, \code{pred} points to a node which at some point in time during its traversal was   the target of a valid search path and was unmarked one level up.
(Intuitively, this assertion justifies starting the new traversal at the middle of the list.)
\item 
The former assertion allows to establish the loop invariant of the intermediate (\code{while(cont)}) loop which says that at \Cref{ln:sl:contains:inv:while:cur} \code{curr} points to a node which at some point in time during the traversal was the target of a valid search path. 
\item 
The role of the  loop invariant of the internal (\code{while(marked)}) loop is key.
It says that whenever the loop is about to start (\Cref{ln:sl:contains:inv:while:marked}), not only was the $level$ entry of the {\tmnext} array of \code{curr} the target of a valid search path, but at that time its successor at level {\code{level}} was \code{succ} and that entry at that level was marked only if the \code{marked} variable is true.  	
\end{compactitem}
We lift the invariants  to concurrent executions  by applying
the extension of our framework discussed in \Cref{rem:ReachWithField}.
From this point on,  it is easy to prove that in \Cref{ln:sl:contains:ret}, \code{curr} points to a node which at some point in time its bottom {\tmnext} entry was unmarked and the target of a search path. As the key of the node is  unmarked with a key equal or greater to the desired one. Linearizability follows.


\lstset{language={program},style=lnumbers,firstnumber=last}

\begin{figure}
\centering
\begin{lstlisting}
bool contains(int v) 
  SLN pred$\leftarrow$head, curr$\leftarrow$null, succ$\leftarrow$null;
  for (int level$\leftarrow$L; 0 $\leq$ level; level--)
    $\{\Past{\slQReachXK{pred.\inext_{level}}{v} \land \neg pred.next[level+1].marked)} \land pred.key < v  \}$ $\label{ln:sl:contains:inv:for:pred}$
    curr$\leftarrow$pred.next[level].ref
    bool cont$\leftarrow$true
    while (cont)
      $\{\Past{\slQReachXK{curr.\inext_{level}}{v}} \}$ $\label{ln:sl:contains:inv:while:cur}$  
      (succ,marked)$\leftarrow$curr.next[level] 
      while (marked)
        $\{\Past{\slQReachXK{curr.\inext_{level}}{v} \land curr.next[level]=(succ,marked)} \}$ $\label{ln:sl:contains:inv:while:marked}$
        curr$\leftarrow$succ 
        $\{\Past{\slQReachXK{curr.\inext_{level}}{v}} \}$ $\label{ln:sl:contains:inv:for:succ}$
        (succ,marked)$\leftarrow$curr.next[level] 
      $\{\Past{\slQReachXK{curr.\inext_{level}}{v} \land curr.next[level]=(succ,marked)} \}$ $\label{ln:sl:contains:after-while}$ 
      if (curr.key < v)
        $\{\Past{\slQReachXK{curr.\inext_{level}}{v} \land  curr.next[level]=(succ,false)} \land curr.key < v  \}$ 
        pred$\leftarrow$curr
        curr$\leftarrow$succ
      else
        $\{\Past{\slQReachXK{curr.\inext_{level}}{v} \land  curr.next[level]=(succ,false)} \land curr.key \geq v  \}$ 
        cont$\leftarrow$false

    $\{\Past{\slQReachXK{curr.\inext_0}{v} \land \neg curr.next[0].marked)} \land curr.key \geq v \}$  $\label{ln:sl:contains:ret}$   
    return curr.key = v   
\end{lstlisting}
\caption{\label{Fi:LFSkipList:Contains}
Lock-free concurrent skiplist algorithm: Procedure \code{contains}. (See~\cite[Fig 14.14]{TAOMPP}).
} 
\end{figure}









\fi

\end{document}